\documentclass[a4paper,11pt]{article}
\usepackage[a4paper, total={6in, 8in}]{geometry}

\usepackage{amsmath,amsthm,amssymb}
\usepackage{graphicx}
\usepackage[font=scriptsize,labelfont=bf]{caption}

\newtheorem{theorem}{Theorem}
\newtheorem{lemma}[theorem]{Lemma}
\newtheorem{corollary}[theorem]{Corollary}

\newtheorem{conjecture}[theorem]{Conjecture}

\theoremstyle{definition}
\newtheorem{remark}[theorem]{Remark}
\newtheorem{example}[theorem]{Example}

\newcommand{\eg}{{e.\,g.}}
\newcommand{\ie}{{i.\,e.}}

\newcommand{\OPT}{\operatorname{opt}} 
\newcommand{\calI}{\mathcal{I}} 
\newcommand{\calJ}{\mathcal{J}} 
\newcommand{\calM}{\mathcal{M}} 

\newcommand{\NP}{\mathsf{NP}} 
\newcommand{\OCNP}{\mathsf{Oracle CoNP}} 

\newcommand{\nats}{\mathbb{N}} 
\newcommand{\reals}{\mathbb{R}} 
\newcommand{\ints}{\mathbb{Z}} 
\newcommand{\nn}{_+} 

\newcommand{\eps}{\varepsilon}

\newcommand{\E}{\mathbb{E}} 
\newcommand{\con}{\operatorname{conv}} 
\usepackage{url}

\begin{document}

\title{Robust randomized matchings\footnote{A preliminary version of this paper appeared in the proceedings of SODA 2015~\cite{matuschke2015robust}.}}

\author{Jannik Matuschke\footnote{TUM School of Management and Department of Mathematics, Technische Universit\"at M\"unchen} \and Martin Skutella\footnote{Institut f\"ur Mathematik, Technische Universit\"at Berlin} \and Jos\'{e} A. Soto\footnote{Departamento de Ingenier\'ia Matem\'atica and CMM, Universidad de Chile}}

\maketitle

\abstract{%
The following game is played on a weighted graph: Alice selects a matching $M$ and Bob selects a number $k$. Alice's payoff is the ratio of the weight of the $k$ heaviest edges of $M$ to the maximum weight of a matching of size at most $k$. If $M$ guarantees a payoff of at least $\alpha$ then it is called $\alpha$-robust. Hassin and Rubinstein~\cite{hassin2002robust} gave an algorithm that returns a $1/\sqrt{2}$-robust matching, which is best possible.

We show that Alice can improve her payoff to $1/\ln(4)$ by playing a randomized strategy. This result extends to a very general class of independence systems that includes matroid intersection, b-matchings, and strong 2-exchange systems. 
It also implies an improved approximation factor for a stochastic optimization variant known as the maximum priority matching problem and translates to an asymptotic robustness guarantee for deterministic matchings, in which Bob can only select numbers larger than a given constant.
Moreover, we give a new LP-based proof of Hassin and Rubinstein's bound.
}

\section{Introduction}
Hassin and Rubinstein~\cite{hassin2002robust} introduced the concept of \emph{robust matchings}, which can be described by the following zero-sum game:
Given a graph \mbox{$G=(V,E)$} with non-negative edge weights $w \colon E \to \reals\nn$, Alice selects a (not necessarily perfect) matching $M$ in $G$. Then Bob, an adversary, selects a bound $k$ on the number of allowed elements. Lastly, Alice outputs the
$k$ heaviest elements of $M$ and receives a payoff equal to the ratio between the weight of the output and the maximum weight of a matching in $G$ of
cardinality at most $k$. She wants to maximize her payoff while Bob wants to minimize it. A matching $M$ that guarantees a payoff of at least~$\alpha$~(for some $\alpha \in [0, 1]$) is called \emph{$\alpha$-robust}.
Hassin and Rubinstein~\cite{hassin2002robust} showed that the $p$-th power algorithm, \ie, the one where Alice returns a maximum matching with respect to
the $p$-th power weights $w^p$, provides a $\min\{2^{-1+1/p}, 2^{-1/p}\}$-robust solution. In particular, by setting $p=2$, this implies the existence
of a $1/\sqrt{2}$-robust matching in any graph. Figure~\ref{fig:tight} shows that this bound is best possible, as the depicted graph admits no matching that is more than $1/\sqrt{2}$-robust.

In this paper, we show that the bound of $1/\sqrt{2}$ can be overcome when allowing Alice to play a mixed strategy, \ie, instead of specifying a single deterministic matching, she announces a probability distribution on the matchings of the graph. Bob then chooses the cardinality based on his knowledge of the distribution. We show that in this natural extension of Hassin and Rubinstein's problem, Alice can always find a distribution that has an expected payoff of at least $1/\ln(4) > 1/\sqrt{2}$.
We complement this result by showing that our guarantee of $1/\ln(4)$ carries over to an asymptotic setting in which Alice specifies a deterministic matching but Bob's choice is restricted to numbers~$k \geq K$ for some constant~$K$. 
Using the minimax theorem for zero-sum games, we also show a close relation between our robust randomized matchings and the maximum priority matching problem, which can be interpreted as a stochastic variant of the robust matching problem.
Our robustness results for randomized matchings imply an improved approximation factor for this problem.
In addition, we also provide a new simple LP-based proof of Hassin and Rubinstein's original result~\cite{hassin2002robust} on the squared weight algorithm.

\begin{figure}[t]
\centering \includegraphics[scale=1]{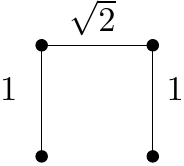}
\caption{Tight example for the existence of a $1/\sqrt{2}$-robust matching. The matching consisting of the single edge of weight~$\sqrt{2}$ is off by a factor of $1/\sqrt{2}$ from the optimum for~$k=2$. On the other hand, the matching consisting of the two edges of weight~$1$ is off by a factor of $1/\sqrt{2}$ from the optimum for~$k=1$.\label{fig:tight}}
\end{figure}

\subsubsection*{Notation}
An \emph{independence system} $\calI$ on a ground set $E$ is a nonempty family of subsets of~$E$, called \emph{independent sets}, such that $\calI$ is downward closed, \ie, $S \subseteq T \in \calI$ implies $S \in \calI$. Let $w\colon E \to \reals\nn$ be a non-negative weight function on the ground set of~$\calI$. For any set $S \subseteq E$ define $w(S) = \sum_{e \in S} w_e$. For any~$k \in \nats$, let~\mbox{$\OPT_k = \max \{w(S) \, : \, S \in \calI, \, |S| \leq k\}$}, and for any independent set $S \in \calI$ denote by $S_k \subseteq S$ the subset of the $k$ heaviest elements in $S$ (breaking ties consistently) and~$S_k:=S$ if~$|S|<k$.  An independent set $S \in \calI$ is $\alpha$-robust with respect to $w$ if and only if
$w(S_k) \geq \alpha \cdot \OPT_k$
for all $k \in \nats$. We also use~\mbox{$\Delta(\calI)=\{\lambda \in [0,1]^\calI \colon \sum_{I \in \calI} \lambda_I = 1\}$} to denote the set of probability distributions over $\calI$. These distributions are also called \emph{randomized independent sets}. In particular, if $\calM$ denotes the set of matchings of a given graph $G$ then the elements of $\Delta(\calM)$ are called \emph{randomized matchings}.

\subsection{Related work}
The idea of finding a combinatorial object (a weighted matching in this case) that is robust against an adversarial choice of cardinality has also
been studied on other domains. Fujita, Kobayashi, and Makino~\cite{fujita2013robust} proved that the above results for matchings hold for the problem of
finding common independent sets of two matroids. They also showed that computing the maximum robustness factor $\alpha$ for a given instance
is $\NP$-hard even for the case of matchings in bipartite graphs. 
For general independence systems, Hassin and Rubinstein~\cite{hassin2002robust}, extending results from~\cite{hausmann1980worst}, observed that the greedy algorithm computes a $\nu$-robust set, where $\nu$ is the rank quotient of the system\footnote{The \emph{rank quotient} $\nu(\calI)$ of independence system $\calI$ is defined as
$\min_{X \subseteq E} \underline{r}(X)/\overline{r}(X)$, where $\underline{r}(X)$ and $\overline{r}(X)$ are the smallest and largest cardinality of a maximal independent set inside $X$, respectively.}.
Kakimura and Makino~\cite{kakimura2013robust} showed that every
independence
system admits a $1/\sqrt{\mu}$-robust solution, where $\mu$ denotes the \emph{extendibility}\footnote{An independence system $\calI$ is \emph{$\mu$-extendible} if for every $X, Y \in \calI$ and $y\in Y\setminus X$ there is $Z \subseteq  X \setminus Y$, $|Z|\leq \mu$ such that $(X \cup \{y\}) \setminus Z \in \calI$. In particular, the extendibility of a system is always at least the inverse of the rank quotient.} of the system, 
and in fact, this solution can be found by the (not necessarily polynomial-time computable) squared weight algorithm. 
Hassin and Segev~\cite{hassin2006robust} considered the problem of finding a small subgraph of a given graph that contains for every $k \in \nats$ a spanning tree (or path, respectively) of cardinality at most $k$ and weight at least $\alpha$ times the weight of a maximum weight solution of size~$k$. They show that $\alpha|V|/(1-\alpha^2)$ edges suffice and give polynomial time algorithms for robust solutions using the results in~\cite{hassin2002robust}.
In a wider context, Hassin and Rubinstein's work inspired other robustness results, \eg, in graph coloring~\cite{fukunaga2008robust}, knapsack problems~\cite{kakimura2012computing,disser2014packing}, and sequencing~\cite{megow2013instance}.
More recently, D\"utting, Roughgarden, and Talgam-Cohen~\cite{duetting2014modularity} discovered an application of Hassin and Rubinstein's analysis of the squared weight algorithm in the design of double auctions. 

Results on the impact of randomization in robust optimization are much scarcer. 
Bertsimas, Nasrabadi, and Orlin~\cite{bertsimas2016power} show that the randomized version of network flow interdiction is equivalent to the maximum robust flow problem and use this insight to obtain an approximation algorithm for both problems. Mastin, Jaillet, and Chin~\cite{mastin2015randomized} consider a randomized version of the minmax (additive) regret model. They show that if an optimization problem can be solved efficiently, then also the corresponding randomized minmax regret version can be solved efficiently. 
Inspired by a preliminary version of our work, Kobayashi and Takazawa~\cite{kobayashi2016randomized} provide an analysis of the randomized robustness of general independence systems. 
They obtain a robustness guarantee of $\Omega(1/\log(\rho))$ for general systems and a guarantee of $\Omega(1/\log(\mu))$ for systems induced by instances of the knapsack problem (here $\mu$ is again the extendibility of the system, while $\rho$ is the ratio of the largest size of an independent set to the smallest size of a non-independent set minus $1$). They also construct instances that do not allow for a better than $O(\log \log (\mu) / \log (\mu))$-robust or $O(\log \log (\rho) / \log(\rho))$-robust randomized independent set, respectively.

\subsection{Our Contribution}
We analyze natural extensions of Hassin and Rubinstein's robust matching game and also present a new proof for the $1/\sqrt{2}$-robustness of the squared weight algorithm. Our main results are the following.

\subsubsection*{Randomized robustness}
We show that Alice can improve on the guarantee of $1/\sqrt{2}$ when allowing her to play a randomized strategy.
For this setting, we provide a simple algorithm that allows Alice to receive an expected payoff of $1/\ln(4)$ by rounding the weights according to a randomized parameter and then computing a lexicographically optimal matching with respect to these rounded weights.
The result is based on the insight that if all weights are powers of $2$, then any lexicographically maximum matching is $1$-robust. In fact, we prove this property for a very general family of independence systems that are characterized by a concave behavior of the optimal solution value $\OPT_k$ depending on the cardinality~$k$. This family forms a strict subset of the 2-extendible systems for which deterministic $\sqrt{2}$-robustness is possible, but includes very interesting cases, such as matroid intersections, $b$-matchings, and strong $2$-exchange systems. 

\subsubsection*{Maximum priority matching} 
Robust matchings provide a concept for mitigating the worst-case deviation from the optimum when the cardinality bound is unknown. In cases where we can assume knowledge of the distribution of the cardinality constraint, an alternative objective is to find a (deterministic) matching maximizing the expected weight of the heaviest $k$ edges, where~$k$ is a random variable with known distribution. This stochastic variant of the problem is known as the maximum priority matching problem. We observe that this problem is essentially equivalent to finding Alice's best response to a randomized strategy played by Bob in the robust matching game. Therefore, our analysis of robust randomized matchings together with the minimax theorem for zero-sum games implies an improved $1/\ln(4)$-approximation factor for the maximum priority matching problem.

\subsubsection*{Asymptotic robustness}
Motivated by the observation that worst-case examples for the best possible robustness guarantee are attained on very small graphs, we consider an asymptotic setting in which Alice specifies a deterministic matching but Bob's choice is restricted to numbers $k \geq K$ for some constant $K$. We show that this setting is connected to the randomized setting in two ways. On the one hand, we show that in both settings, Alice's payoff cannot exceed a value of $(1 + 1/\sqrt{2})/2$. On the other hand, we give a general technique for transforming any distribution on a constant number of matchings into a deterministic matching while attaining almost the same robustness guarantee as the former for large cardinalities. We use this to show that Alice can always obtain a payoff of $1/\ln(4) - \eps$, where $\eps > 0$ depends only on $K$ and becomes arbitrarily small as $K$ increases.

\subsubsection*{An LP-based proof for the squared weight algorithm}
Hassin and Rubinstein's proof of the $\sqrt{2}$-robustness of matchings computed by the squared weight algorithm~\cite{hassin2002robust} is based on an optimization problem whose variables are the edge weights of worst-case instances for the algorithm.
The proof involves a sequence of technical arguments for showing that the worst case is attained for graphs with $3$ or $5$ edges.
Fujita, Kobayashi, and Makino's~\cite{fujita2013robust} generalization of this result also makes use of an optimization problem over the weights, additionally using the dual of the matroid intersection polytope to split the weight of each edge in two separate parts.
We give a new proof of Hassin and Rubinstein's original result based on LP duality.
In contrast to both proofs mentioned above, our proof does not involve deducing properties of the worst case instance. Instead, we give an explicit construction of a feasible dual solution to the $k$-matching LP for each cardinality $k$, using an optimal dual solution of the matching LP with squared weights.

\section{Randomized robustness}\label{sec:randomized}

Consider the natural extension of the zero-sum game in which Alice can play a mixed strategy, \ie, instead of choosing a single matching $M$, Alice now specifies a distribution~\mbox{$\lambda \in \Delta(\mathcal{M})$} on the set $\mathcal{M}$ of matchings of the graph. Then Bob announces a cardinality $k \in \nats$ and Alice receives the payoff $\E_{M \sim \lambda}[w(M_k)]/\OPT_k$.

A distribution $\lambda$ of matchings is called \emph{$\alpha$-robust} with respect to $w$ if and only if
\[\E_{M \sim \lambda}[w(M_k)] \ \geq \ \alpha \cdot \OPT_k \qquad \text{for all } k \in \nats.\]
The \emph{randomized robustness ratio} $\alpha^*(G)$ of the graph $G$ is the largest value $\alpha^*$ such that for any weight function, there exists an $\alpha^*$-robust
distribution. Naturally,
\[1/\sqrt{2} \leq \alpha(G) \leq \alpha^*(G) \leq 1\]
for all graphs $G$.
In this section, we will show that Alice can always construct a $1/\ln(4)$-robust randomized matching. In fact, our result holds in a much more general setting of a larger class of independence systems, which we call \emph{bit-concave}.

Before we turn our attention to this positive result, we also give an upper bound on the best randomized robustness factor that can be guaranteed on arbitrary graphs. Interestingly, the currently best-known upper bound stems from the same worst-case example as for the deterministic setting.
\medskip

\begin{example}
	For the graph $G$ given in Figure~\ref{fig:tight}, $\alpha^*(G) = (1 + 1/\sqrt{2})/2$. 
	Assume that the optimal randomized matching $M^*$ chooses the matching containing only the central edge with probability $p$ and chooses the matching with the two outer edges with probability~\mbox{$1-p$} (note that we can assume without loss of generality that only maximal matchings appear in the support of the distribution). Then $\E[w(M^*_1)]/\OPT_1 = p + (1-p)/\sqrt{2}$ and $\E[w(M^*_2)]/\OPT_2 = p/\sqrt{2} + (1 - p)$. Therefore the robustness of the matching is \begin{align*}
	\min \{p + (1 - p)/\sqrt{2},\ p/\sqrt{2} + (1 - p)\},
\end{align*} which is maximized for $p = 1/2$.
\end{example}
\medskip

Note that the notion of randomized robustness naturally generalizes to independence systems by simply replacing matchings with independent sets, \ie, we replace $\calM$ by an independence system~$\calI$. We will use this more general notation throughout the remainder of this section.

\subsubsection*{Bit-functions and good systems}
Let $(E, \calI)$ be an independence system and $w\colon E \to \reals\nn$ be  a weight function. 
We define the \emph{lexicographic order} on $\calI$ with respect to $w$ as follows.
For two sets $S, T \in \calI$ consider the two sequences $(w_e)_{e\in S}$ and $(w_e)_{e\in T}$, respectively, each sorted in order of decreasing weight. We say that $S$ is \emph{$w$-lexicographically larger} than $T$ if $(w_e)_{e\in S}$ is lexicographically larger than $(w_e)_{e\in T}$.
A set is \emph{$w$-lexicographically maximal} if no other set is $w$-lexicographically larger than it. 
Observe that for every weight function $w$, there is a constant $C > 0$ such that if we define $w^*_e := w_e^C$ then for any pair of sets $S$, $T \in \calI$, we have $w^*(S)\geq w^*(T)$ if and only if $S$ is $w$-lexicographically larger than $T$. Instead of $w$-lexicographically maximal sets, we will therefore refer to $w^*$-maximal sets in the following. 
In particular, if we can compute a maximum weight independent set for a system in polynomial time, we can also find a $w^*$-maximal independet set in polynomial time\footnote{Note that we can assume $w_e \in \{1, \dots, |E|\}$ for all $e \in E$ w.l.o.g.\ since we are looking for a lexicographically maximum solution. Therefore, it is sufficient to choose $C = |E|$ and the weights $w^*$ can thus be assumed to be of polynomial encoding size.} (note however that simply running the greedy algorithm does not necessarily yield the desired result, as illustrated in Figure~\ref{fig:lex}).

\begin{figure}[ht]
	\centering \includegraphics[scale=1]{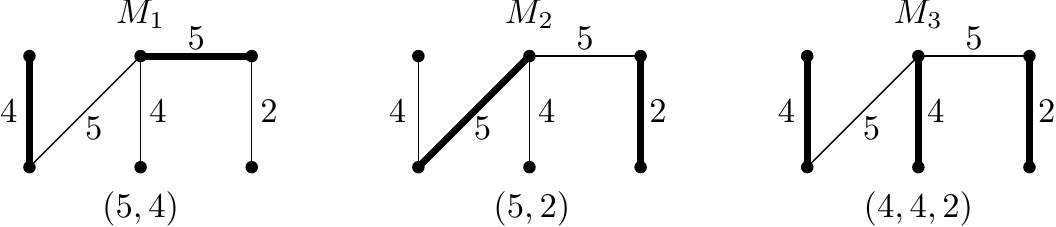}
	\caption{The figure shows three maximal matchings of a given graph on 6 vertices and their corresponding sequences of weights. Observe that $w^*(M_1)>w^*(M_2)>w^*(M_3)$, $w(M_3)>w(M_1)>w(M_2)$, $M_1$ is the $w^*$-maximal matching, and $M_3$ is the $w$-maximal one. Both $M_1$ and $M_2$ are possible outcomes of the greedy algorithm.}\label{fig:lex}
\end{figure}

A function $w\colon E \to \reals\nn$ is called a \emph{bit-function} if $w_e=2^{\ell_e}$ for all $e\in E$, with $\ell_e \in \ints$. A system $\calI$ is called \emph{good} if, for every bit function $w$, the $w^*$-maximal independent sets are 1-robust for $w$.
We will later give a characterization of good systems that comprises a very general family of independence systems, including the set of matchings or $b$-matchings of a graph, the intersection of two matroids, and strong 2-exchange systems. See Figure \ref{fig:good} for a matching example.

\begin{figure}[ht]
	\centering \includegraphics[scale=1]{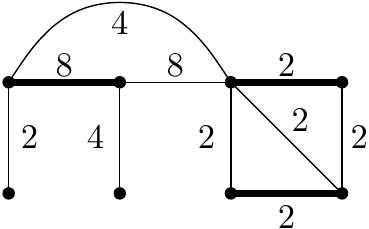}
	\caption{Instance with a bit-function weighting $w$. One $w^*$-maximal matching $L$ is shown with thick lines. The matching $L$ is $1$-robust for $w$ since $w(L_1)=\OPT_1=8$, $w(L_2)=\OPT_2=10$, and for every $k\geq 3$, $w(L_k)=\OPT_k=12$.}\label{fig:good}
\end{figure}

\medskip
The remainder of this section is organized as follows. In Section~\ref{sec:ln4-algorihm}, we give a simple algorithm for computing a $1/\ln(4)$-robust independent distribution in any good system. We then provide an alternative characterization of good systems in Section~\ref{sec:bit-concave}, showing that for every bit function in a good system, the value $\OPT_k$ is a concave function of the cardinality $k$. In Section~\ref{sec:good-examples} we then give several examples of good systems. Finally, in Section~\ref{sec:2-extendible} we show that good systems are a strict subclass of $2$-extendible systems.

\subsection{An algorithm for computing $1/\ln(4)$-robust independent distributions in good systems}\label{sec:ln4-algorihm}

Our interest in good systems is motivated by the insight that we can round arbitrary weight functions to powers of~$2$, losing only a factor of $1/\ln(4)$ in expectation.
We exploit this fact in the algorithm described in the proof of the following theorem.

\begin{theorem}\label{thm:ln4-robust}
If $\calI$ is a good system, then $\alpha^*(\calI)\geq 1/\ln(4)$.
\end{theorem}
\begin{proof}
The idea is to round all weights to powers of $2$ in a randomized fashion.
Without loss of generality we can assume $w$ to be strictly positive (by replacing $0$ weights by a very small $\varepsilon > 0$).
Define for every element $e \in E$, and every $x\in [0,1]$ the following values:
\begin{align*}
\ell_e & \; := \; \log_2(w_e),\\
\ell^{[x]}_e & \; := \; \lfloor \ell_e-x\rfloor,\\
w^{[x]}_e & \; := \; 2^{\ell^{[x]}_e} \; = \; w_e \cdot 2^{\ell^{[x]}_e-\ell_e}.
\end{align*}

For every $x$, the function $w^{[x]}\colon E \to \reals\nn$ is a bit-function. Since $\calI$ is a good system, there exists a feasible set $S^{[x]} \in \calI$ that is 1-robust for $w^{[x]}$. We construct a randomized solution to the instance $(\calI, w)$ by using the following algorithm: Select $x$  from the uniform distribution $U[0,1]$ and return $S^{[x]}$.

We now show that the above algorithm returns a distribution that is $1/\ln(4)$-robust. Given $k$, let $S^*$ be the feasible set in $\calI_k:=\{S\in\calI\, : \,|S|\leq k\}$ that maximizes $w$. We have
\begin{align*}
\E[w(S^{[x]}_k)] & \ \geq \ \E[2^x \cdot w^{[x]}(S^{[x]}_k)] \ \geq \ \E[2^x \cdot w^{[x]}(S^*)]\\
&\ = \, \sum_{e\in S^*} \E[2^x \cdot w^{[x]}_e]\ =\, \sum_{e\in S^*} w_e \, \E[2^{\ell^{[x]}_e - (\ell_e - x)}] \\
&\ = \ w(S^*) \cdot \E_{u\sim U[0,1]}[2^{-u}]\ =\ \frac{\OPT_k}{\ln(4)},
\end{align*}
where the first inequality follows from the fact that $w_e = w^{[x]}_e2^{\ell_e- \lfloor \ell_e-x\rfloor}\geq 2^xw^{[x]}_e$ for every $e$, the second inequality holds since $S^{[x]}_k$ is 1-robust for $w^{[x]}$, and the penultimate equality follows from the fact that $f : [0, 1] \rightarrow [-1, 0]$ defined by $f(x) := \lfloor\ell - x\rfloor - (\ell - x)$ preserves the uniform measure.
This concludes the proof.
\end{proof}

\begin{corollary}\label{cor:computation}
A $1/\ln(4)$-robust independent distribution can be explicitly found by solving $|E|$ instances of the maximum weight independent set problem on $\calI$.
\end{corollary}

\begin{proof}
Let~$q_e:=\lfloor\log_2(w_e)\rfloor$. The rounded weight of element~$e$ is~$w^{[x]}_e=2^{q_e}$ if~$x\leq \log_2(w_e)-\lfloor\log_2(w_e)\rfloor$ and~$w^{[x]}_e=2^{q_e-1}$ otherwise. Thus, there are at most~$|E|+1$ different bit-functions~$w^{[x]}$ for different values of $x \in [0,1]$.
\end{proof}

\begin{remark}\label{rem:tight-example}
The analysis of the algorithm described in the proof of Theorem~\ref{thm:ln4-robust} is tight. To see this, consider the following weighted graph. Let $n \in \nats$. There are $2^n$ vertices $v_0, \dots, v_{2^n - 1}$ and $n$ different types of edges. For $k \in \{0, \dots, n-1\}$, there are $2^k$ edges of type $k$, each of the form $\{v_i, v_{i + 2^k}\}$ for $i \in \{0, \dots, 2^k - 1\}$. All edges of type $k$ have weight $2^{(n - k)/n}$.
Note that for edge $e$ of type $k$, the randomized weight is $w^{[x]}_e = 1$ if $x \leq (n - k)/n$ and $w^{[x]}_e = 1/2$ otherwise. Therefore, the $w^*$-maximal matching $S^{[x]}$ maximizes the number of edges with $w^{[x]}_e = 1$, all of which are of type at most $k^* := \lfloor (1 - x)n \rfloor$. Note that this number is maximized by including all edges of type exactly $k^*$ and no edges of type $k < k^*$. We conclude that $\E[w(S^{[x]}_1)]/\OPT_{1} = \frac{1}{2} \sum_{k = 0}^{n - 1} \operatorname{Pr}(\lfloor (1-x)n \rfloor = k) \cdot 2^{\frac{n - k}{n}} = \frac{1}{n} \sum_{k = 0}^{n - 1} 2^{-\frac{k}{n}}$
and observe that this value converges to $1/\ln(4)$ as $n$ goes to infinity.
\end{remark}

\subsection{Good systems are bit-concave}\label{sec:bit-concave}
In order to give better characterizations of \emph{good systems}, we need to recall some definitions.

\subsubsection*{Deletions, contractions, truncations, and t-minors}
Let $\calI \subseteq 2^E$ be an independence system. For all $X \subseteq E$ we define the \emph{deletion} of $X$ in $\calI$ as the system
$\calI \setminus X = \{S \in \calI \, : \, S \cap X = \emptyset\}$.
Furthermore, if $X \in \calI$, we define the \emph{contraction} of $X$ in $\calI$ as the system $\calI / X = \{S \in E\setminus X \, : \, S \cup X \in \calI\}$. For $k \in \nats$ the \emph{$k$-truncation} of $\calI$ is $\calI_k = \{S \in \calI \, : \, |S| \leq k\}$. The elements of $\calI_k$ are called $k$-independent sets. A \emph{truncation-minor}, or simply \emph{t-minor} of $\calI$ is any system that can be obtained from $\calI$ via iterative deletions, contractions or truncations.

\subsubsection*{Bit-concave systems}

An independence system~$\calI$ is called \emph{bit-concave} if, for every bit-function~$w$, the function $\OPT : \nats \rightarrow \reals\nn$ is concave, \ie, $\OPT_k + \OPT_{k+2} \leq 2\cdot \OPT_{k + 1}$ for all $k \in \nats$.\footnote{Note that this inequality is trivially true for $k = 0$, since $\OPT_0 = 0$ and $\OPT_2 \leq 2 \OPT_1$ in any independence system.}\\

The main result of this part states that good systems and bit-concave systems are equivalent.

\begin{theorem}\label{thm:bit-concave}
For an independence system~$\calI$, the following are equivalent:
\begin{enumerate}
\item[(i)] $\calI$ is bit-concave.
\item[(ii)] All t-minors of $\calI$ are bit-concave.
\item[(iii)] For every weight function~$w$ and every t-minor $\calJ$ of $\calI$, the $w^*$-optimal feasible sets in $\calJ$ are also $w$-optimal.
\item[(iv)] $\calI$ is a good system.
\end{enumerate}
\end{theorem}
\begin{proof} 
\noindent\textbf{[(i) $\Rightarrow$ (ii)]} Let $\calI$ be bit-concave and let $\calI'$ be a t-minor of $\calI$. For $w': E \rightarrow \reals\nn$, we denote by $\OPT'_k := \max \{w(S) \, : \, S \in \calI', |S| \leq k\}$ the weight of an optimal set of cardinality $k$ in $\calI'$.

We first consider the case that $\calI'$ is the $t$-truncation of $\calI$. Indeed, for all bit-functions~$w$ we have $\OPT'_k = \OPT_k$ if $k\leq t$ and $\OPT'_k = \OPT_t$ if $k > t$. Since $\OPT_k$ is concave, so is $\OPT'_k$.

Let us now check that also deletions are bit-concave. Let $X\subseteq E$, $\calI'=\calI\setminus X$ and $w'\colon E(\calI \setminus
X)\to \reals$ be a bit-function on $E\setminus X$. We define
\[\Delta(S, T, U):=\frac{w'(S)+w'(T)}{2} - w'(U)\]
for every triple $S, T, U \in \calI'$ and we further let $$\delta := \inf \{\Delta(S, T, U) \, : \, S, T, U \in \calI' \text{ and } \Delta(S, T, U) > 0\}.$$ Let $\varepsilon$ be any integer power of $2$ strictly smaller than $\delta/|X|$ and extend $w'$ by defining $w\colon E\to \reals$ as
$w(e) = w'(e)$ if $e\in E\setminus X$, and $w(e) = \varepsilon$ if $e \in X$.
We obtain
\begin{align*}
\frac{1}{2}(\OPT'_{k+2}+\OPT'_k) - \OPT'_{k+1}&\leq \frac{1}{2}(\OPT_{k+2}+\OPT_{k}) - \OPT'_{k+1}\\
&\leq \OPT_{k+1}-\OPT'_{k+1}\\
&\leq |X|\varepsilon < \delta.
\end{align*}
If the left hand side was positive, then it would be at least $\delta$, which is a contradiction. We conclude that the left hand side is non-positive
and so $\OPT'$ is concave.

Finally, consider the case that $\calI' = \calI/X$ is a contraction of $\calI$ for some $X\in \calI$. Let \mbox{$w'\colon E(\calI / X)\to \reals$} be a bit-function on $E\setminus X$. Let $M:=2^{|E|}\cdot\max_{e \in
E\setminus X} w'(e)$ and define $w\colon E\to \reals$ as
$w(e) = w'(e)$ if $e\in E\setminus X$ and $w(e) = M$ if $e \in
X$.
Observe that if $i\leq |X|$ then $\OPT_i = Mi$, and
if $i \geq |X|$ then $\OPT_i=M|X|+\OPT'_{i-|X|}$. Since $\OPT$ is concave we deduce that $\OPT'$ is concave as well.\\[5pt]
\noindent\textbf{[(ii) $\Rightarrow$ (iii)]} We will prove this claim by induction. More precisely, let $\calJ$ be a t-minor of $\calI$ and assume that
the claim holds for every strict t-minor of $\calJ$. The (trivial) base cases are when $\calJ$ contains 0 or 1 independent sets. Let $A$ be a $w^*$-maximum
feasible set of $\calJ$ and suppose by contradiction that there is a set $B \in \calJ$ with higher $w$-value.

\medskip
\noindent\emph{Observation 1:} If $w(B)>w(A)$ then $B\cap A =\emptyset.$
Otherwise, consider the contracted system (and strict t-minor) $\calJ'=\calJ/(B\cap A)$. $A\setminus B$ is $w^*$-maximum
in $\calJ'$. By induction $w(A\setminus B)\geq w(B\setminus A)$ and so, $w(A)\geq w(B)$.

\smallskip
\noindent\emph{Observation 2:} If $w(B)>w(A)$ then \mbox{$B\cup A =E(\calJ)$}.
To see this, assume $B\cup A \neq E(\calJ)$ and consider the strict t-minor \mbox{$\calJ'=\calJ\setminus( E\setminus (B\cup A))$}. $A$ is $w^*$-maximum in $\calJ'$, and so $w(A)\geq w(B)$.

\medskip
By Observations 1 and 2 we conclude that $B=E\setminus A$ is the unique feasible set in $\calJ$ with $w(B)>w(A)$. Now let $a$ be a smallest element in $A$ and let $C$ be a lexicographically maximum feasible set in $\calJ\setminus\{a\}$. By induction, the lexicographic maximality of $C$ implies $w(C) \geq w(B) > w(A)$ and thus $C = B$.
Therefore in particular
$w^*(A \setminus \{a\}) \leq w^*(B) < w^*(A)$ and $w(A \setminus \{a\}) \leq w(A) < w(B)$.

Let $k := |A|$. Since $B$ is lexicographically between $A \setminus \{a\}$ and $A$, we conclude that $B_{k-1}$ is lexicographically equal to $A \setminus \{a\}$, \ie,
$w^*(B_{k-1})=w^*(A \setminus \{a\})$. Let $\bar{B}=B \setminus B_{k-1}$. We have
\begin{align*}
w(a) & \ = \ w(A)-w(A \setminus \{a\}) \ = \ w(A)-w(B_{k-1}) \ <\ w(B)-w(B_{k-1}) \ = \ w(\bar{B}).
\end{align*}
But then,
$w^*(\bar{B}) < w^*(a) \;\text{ and }\; w(\bar{B}) > w(a)$.
This means that $\bar{B}$ is formed by elements whose individual weights are less than $w(a)$, but in total, they sum to a weight larger than~$w(a)$.
Since all the weights are powers of two, we deduce that $|\bar{B}|\geq 3$. Therefore,
\[\ell := |B|=|B_{k-1}|+|\bar{B}|\geq k+2.\]
To conclude the proof, note that since $B$ is the only feasible set of weight larger than $A$, 
\[w(A)=\OPT_{k} = \OPT_{k+1}<\OPT_{\ell},\]
which contradicts the bit-concavity of $\calJ$.\\[5pt]

\noindent\textbf{[(iii) $\Rightarrow$ (iv)]} Let $w$ be a bit-function, let $A$ be a $w^*$-optimum, and let $k\in \nats$. By
definition of lexicographic maximum we deduce that $A_k$ is also $w^*$-optimum in~$\calI_k$. Using $(iii)$, we conclude that $A_k$ is also
$w$-optimal in~$\calI_k$. As $k$ was chosen arbitrarily,~$A$ is 1-robust.\\[5pt]

\noindent\textbf{[(iv) $\Rightarrow$ (i)]} Let $\calI$ be a good system, $w$ a bit-function on $E(\calI)$, and $A=\{a_1,\ldots,a_\ell\}$ a
$w^*$-maximal set, where $w_{a_1}\geq w_{a_2} \geq \dots \geq w_{a_\ell}$. Since $\calI$ is a good system, for every $k$, $\OPT_k = w(\{a_1,\ldots, a_k\})$. It
follows that $\OPT_{k+2}-\OPT_{k+1}=w_{a_{k+1}} \leq w_{a_k}=\OPT_{k+1}-\OPT_k$ and therefore $\calI$ is bit-concave.
\end{proof}

\subsection{Examples of good systems}\label{sec:good-examples}

As a consequence of Theorems~\ref{thm:ln4-robust} and~\ref{thm:bit-concave}, any bit-concave system admits a $1/\ln(4)$-robust distribution. We now point out some examples of such systems. All of these are in fact \emph{concave} systems, \ie, systems for which the function $\OPT_k$ is concave in $k$ not only for bit-functions but for all nonnegative weight functions.
We are not aware of natural systems that are bit-concave but not concave.\footnote{The conference version of~\cite{kobayashi2016randomized} claimed that systems arising from unit density \textsc{Knapsack} instances fulfill this property. Unfortunately, this claim turned out to be wrong and it no longer appears in the journal version.}

The following lemma characterizes concave systems. A more general statement of this fact was proven by Calvillo~\cite{calvillo1980concavity}.

\begin{lemma}[Calvillo~(1980)~\cite{calvillo1980concavity}]\label{lem:trunc-polytope}
Let $\calI$ be an independence system. $\calI$ is concave \mbox{if and only if}
\begin{gather*}
\con(\{\chi_{S} \, : \, S \in \calI\})\cap \{x \in \reals\nn^E \, : \, x(E)\leq k\}
\ = \
 \con\{\chi_{S} \, : \, S \in \calI_k\}
\end{gather*}
for all  $k \in \nats$.
\end{lemma}

\subsubsection*{Matchings}
Let $\calM$ be the set of matchings on $G$. Instead of using Lemma~\ref{lem:trunc-polytope} (whose hypotheses are not hard to verify), we show directly that $\OPT$ is a concave function for any weight function~$w$.
 Indeed, let $k \in \nats$ and let $M$ be an optimal $k$-matching and $M'$ be an optimal $k+2$-matching. Observe that if $|M| \geq |M'|$, then  $\OPT_k=\OPT_{k+1}=\OPT_{k+2}$. Otherwise, if \mbox{$|M|<|M'|$}, let $P$ be a component of the symmetric difference $M\Delta M'$ with~\mbox{$|P\cap M'|>|P\cap M|$}, i.e., $P$ is an alternating path of $M$ and $M'$ starting and ending at vertices not covered by $M$. Consider the matchings
\mbox{$M_1 = M\Delta P$} and \mbox{$M_2=M'\Delta P$} resulting from swapping the edges along the path $P$ from one matching to the other. Note that $$w(M) + w(M') = w(M_1) + w(M_2).$$ As $|M_1| = |M| + 1 \leq k + 1$ and $|M_2| = |M'|-1 \leq k+1$, we conclude that $$\OPT_k + \OPT_{k+2} = w(M_1) + w(M_2) \leq 2 \OPT_{k+1}.$$

\subsubsection*{Matroid intersection} Given two matroids $\calI_1, \calI_2 \subseteq 2^E$ consider their intersection $\calI=\calI_1\cap
\calI_2$. Let $r_1$ and $r_2$ be the rank functions of $\calI_1$ and $\calI_2$, respectively. It is known~\cite[Section 41.4]{Schrijver} that
\begin{gather*}
\con(\{\chi_I\colon I \in \calI\}) = \{x \in \reals_+^E \,\colon\, x(S)\leq r_i(S)\ \ \forall S\subseteq E, i\in\{1,2\}\}
\end{gather*}
and $\calI^k=\calI_1^k\cap \calI_2^k$, where $\calI_i^k$ is the family of independent sets of the $k$-truncated matroid with rank function $r_i^k(X)=\min \{r_i(X), k\}$. Hence
\begin{align*}
\con(\{\chi_I \, : \, I \in \calI^k\})
& \; = \; \{x \in \reals_+^E \, : \, x(S)\leq r_i^k(S) \ \ \forall S\subseteq E, i\in\{1,2\}\}\\
&\; = \; \con(\{\chi_I \, : \, I \in \calI\}) \cap \{x \in \reals\nn^E \, : \, x(S)\leq k\}.
\end{align*}
Thus, any intersection of two matroids is concave by Lemma~\ref{lem:trunc-polytope}.

\subsubsection*{Strong $2$-exchange systems (including $b$-matchings, SBO matchoids, and SBO matroid parity systems)} 
 
An independence system $\mathcal{I}$ is a \emph{strong $k$-exchange system} if for all $X$, $Y \in\mathcal{I}$, there is a bipartite graph $G(X,Y)$, called $k$-exchange graph from $X$ to $Y$, with color classes $X\setminus Y$ and $Y\setminus X$ of maximum degree at most $k$ such that
$$A\cup (Y\setminus N_{G(X, Y)}(A)) \in \mathcal{I} \quad \text{for all } A \subseteq X\setminus Y,$$
where $N_G(A)$ denotes the neighborhood of $A$ in $G$, \ie, all vertices in $V \setminus A$ adjacent to at least one vertex in $A$.

Strong $k$-exchange systems were introduced by Ward~\cite{ward2012oblivious} in the context of studying non-oblivious local search algorithms for constrained monotone submodular function maximization. We show that all strong $2$-exchange systems are concave and point out some interesting examples of such systems afterwards.

\begin{theorem}
Strong  $2$-exchange systems are concave.
\end{theorem}

\begin{proof}
Observe that the defining condition for strong $2$-exchange systems implies that all the stable sets of $G = G(X,Y)$ belong to $\mathcal{I}/(X\cap Y)$. Indeed, for every stable set $Z\subseteq X\Delta Y$ of $G$,  $N_G(Z\cap X)\subseteq Y\setminus (X\cup Z)$, and so
$$(X\cap Y) \cup Z \subseteq (Z\cap X) \cup (Y \setminus N_G(Z\cap X)) \in \calI.$$
Let $\bar{G}$ be the bipartite graph obtained from \mbox{$G$} by adding pendant edges to all vertices in $X\Delta Y$ of degree less than 2, until all vertices in $X\Delta Y$ have degree 2. The stable sets of $G$ are in correspondence with  the matchings of the line graph of~$\bar{G}$. Therefore, the family of stable sets of $G$ form a concave system. We  use this to show that $\calI$ is concave. Indeed, let $X$ be an optimal $k$-independent set and $Y$ be an optimal $(k+2)$-independent set in $\calI$ with respect to a function $w$.  If $|Y|\leq |X|+1$ then \mbox{$\OPT_{k+1}=\OPT_{k+2}$} and therefore, $\OPT_{k}+\OPT_{k+2}\leq 2\OPT_{k+1}$. Suppose then that $|Y\setminus X|-|X\setminus Y|=|Y|-|X|\geq 2$, let $k' = |Y\setminus X| = k + 2 - |X \cap Y|$ and let $Z$ be a stable set of $G=G(X,Y)$ of size at most $k'-1$ of maximum weight. Since $Y\setminus X$ is an optimal $k'$-stable set and $X\setminus Y$ is an optimal ($k'-2$)-stable set, the concavity of the stable sets of $G$ implies that
\begin{align*}
2\OPT_{k+1} &\geq 2w((X\cap Y)\cup Z)=2w(X\cap Y) + 2w(Z)\\ &\geq 2w(X\cap Y) + w(X\setminus Y)+w(Y\setminus X)\\
& = w(X) + w(Y) = \OPT_k + \OPT_{k+2},  
\end{align*}
and therefore,  $\calI$ is concave.
\end{proof}

In the appendix, we show that strong $2$-exchange systems include a number of interesting special cases: In fact, $b$-matchings, strongly base orderable matchoids, and strongly base orderable matroid parity systems.  
This also implies that strong 2-exchange systems are in general hard to handle algorithmically: computing a maximum weight independent set of an SBO matroid parity system is $\NP$-hard and not even in $\OCNP$~\cite{soto2014mm}. However, it is possible to obtain a PTAS for weighted SBO matroid parity based on local search techniques~\cite{soto2014mm}.

\subsection{Good systems and $2$-extendible systems}\label{sec:2-extendible}
Mestre~\cite{mestre06greedy} defined $\mu$-extendible systems as a natural generalization of matroids on which the greedy algorithm yields a $1/\mu$-factor approximation. Formally\footnote{In the original definition by Mestre, an independence system $\calI$ is \emph{$\mu$-extendible} if for every $C\in \calI$ and $x\not \in C$ such that $C\cup \{x\} \in \calI$ and for every extension $D$ of $C$ there is a subset $Y\subseteq D\setminus C$ with $|Y|\leq \mu$ such that $D\setminus Y \cup \{x\}\in \calI$. Kakimura and Makino noted that both definitions are equivalent.}, an independence system $\calI$ is \emph{$\mu$-extendible} if for every $X, Y \in \calI$ and $y\in Y\setminus X$ there is $Z \subseteq  X \setminus Y$, $|Z|\leq \mu$ such that $(X \cup \{y\}) \setminus Z \in \calI$.
Kakimura and Makino~\cite{kakimura2013robust} showed that every $\mu$-extendible 
system admits a $1/\sqrt{\mu}$-robust (deterministic) solution which can be found by
the (not necessarily polynomial-time computable) squared weight algorithm. Matroid intersection, matchings, $b$-matchings, and matroid parity systems are all examples of $2$-extendible systems. This leads to the natural question how $2$-extendible systems are related to good systems. In fact, we show that good systems are a strict subclass of $2$-extendible systems.

\begin{lemma}\label{lem:good-extendible}
Every good system is $2$-extendible.
\end{lemma}

\begin{proof}
Let $\calI$ be a good system, and let $X, Y\in \calI$ and $y \in Y\setminus X$. 
Let $k = |X \setminus Y|$ and assume that $k\geq 2$ as otherwise proving $2$-extendibility is trivial.
Consider the contraction $\calI' = \calI / (X \cap Y)$, which is bit-concave by Theorem~\ref{thm:bit-concave}.
Consider the bit-function  $w\colon E(\calI') \to \reals$ given by
$$w_e = \begin{cases} 1 &\text { if $e\in X \setminus Y$,}\\ 2&\text{ if $e=y$,}\\ 2^{-M} &\text{ otherwise}\\
\end{cases}
$$
for some large $M\geq k$.

Let $\lambda = \frac{1}{k-1}$. Using that $k-1= 1 \lambda + k(1-\lambda)$, that $\OPT_1=w_y=2$, that $w(X \setminus Y) = k \leq \OPT_k$, and the bit-concavity of $\calI'$, we conclude that
\begin{align*}
\OPT_{k-1} & \ \geq \ \lambda \OPT_1 + (1-\lambda) \OPT_k  \\
& \ \geq \ 2\frac{1}{k-1} + k \frac{k-2}{k-1} \ = \ \frac{1}{k-1} + k-1.
\end{align*}

Let $U \in \calI'_{k-1}$ with $w(U)=\OPT_{k-1}$. Since $M$ was chosen large enough, $U$ must be a subset of $(X \setminus Y) \cup \{y\}$ and so $w(U)$ is an integer larger than or equal to $k$. Since $|U|\leq k-1$, we conclude that $U$ contains $y$ and at least $k-2$ elements from $X \setminus Y$. Let $Z = (X \setminus Y) \setminus U$.
Note that $U \in \calI'$ implies $(X \cup \{y\}) \setminus Z = U \cup (X \cap Y) \in \calI$. This proves that $\calI$ is 2-extendible.
\end{proof}

\begin{lemma}\label{lem:good-extendible-counter-example}
There exist $2$-extendible systems that are not good systems.
\end{lemma}

\begin{proof}
Let $A=\{a_1,a_2\}$ and $B=\{b_1,b_2,b_3,b_4\}$ be two disjoint sets. Define also $C=\{a_1,b_3,b_4\}$ and $D=\{a_2,b_3,b_4\}$. Consider the independence system $\calI^*$ with bases $\{ A, B, C, D\}$. We show that $\calI^*$ is $2$-extendible but not good.

Let $X$ and $Y$ be two bases of $\calI^*$. If  $|X\setminus Y|\leq 2$, the property of extendibility holds by setting $Z=X\setminus Y$. So suppose that $|X\setminus Y|\geq 3$. This can only happen if $X=B$ and $Y=A$. But in this case we can set $Z=\{b_3,b_4\}$ and so $J-Z + a_i \in \{C,D\} \subset \calI^*$. This proves that $\calI^*$ is 2-extendible.

To see that the system is not good, consider the bit-function $$w_e=\begin{cases} 2 &\text{ if $e \in \{a_1,a_2,b_1\},$}\\ 1 &\text{ if $e \in \{b_2,b_3,b_4\}.$}\end{cases}$$
The unique lexicographically optimal ($w^*$-maximal) independent set is $A$. However $A$ is not $w$-maximal, since $w(A)=4<5=w(B)$.
\end{proof}

\section{Stochastic optimization and connections to maximum priority matching}\label{sec:max-priority}
The goal of robust matchings, as considered in the preceding section, is to be as good as possible even in the worst-case. Alternatively, one can consider a \emph{stochastic} setting, in which we have some information about the distribution of the cardinality bound. This problem is known as \emph{maximum priority matching}\footnote{In fact, the original definition of maximum priority matching given by Hassin and Rubinstein~\cite{hassin2002robust} is different, but we show it to be equivalent at the end of this section.}~\cite{hassin2002robust}: Given a probability distribution $\mu \in \Delta(\nats)$ over the natural numbers, find a (deterministic) matching $M \in \calM$ maximizing
\begin{align*}
  \E_{k \sim \mu}[w(M_k)] = \sum_{k \in \nats} \mu_k w(M_k),
\end{align*}
\ie, the expected weight of the best $k$ elements when $k$ is drawn from the distribution $\mu$.
Hassin and Rubinstein~\cite{hassin2002robust} showed that the maximum priority matching problem is $\NP$-hard. They also observed that any $\alpha$-robust matching is an $\alpha$-approximate maximum priority matching. Hence their squared weights algorithms is a $1/\sqrt{2}$-approximation algorithm to the problem. 

In this section, we show that any $\alpha$-robust randomized matching contains an $\alpha$-approximation to maximum priority matching in its support.
In fact, this result again extends directly to priority maximization over independence systems and we will thus again replace $\calM$ by an independence system $\calI$.

\begin{theorem}\label{thm:max-priority}
  Let $\lambda \in \Delta(\calI)$ be the probability distribution corresponding to an $\alpha$-robust randomized independent set. Then for any $\mu \in \Delta(\nats)$ there is an $S \in \calI$ with $\lambda_S > 0$ such that
 \begin{align*}
  \E_{k \sim \mu}[w(S_k)] \geq \alpha \cdot \max_{S^* \in \calI} \, \E_{k \sim \mu}[w(S^*_k)].
\end{align*}
\end{theorem}

\begin{proof}
Let $K := \{|S| \,:\, S \in \calI\}$. Note that we can assume~$\{k \;:\; \mu_k > 0\} \subseteq K$ without loss of generality. 
Furthermore, let $\calI' := \{S \in \calI \,:\, \lambda_S > 0\}$ and consider the following primal/dual pair of linear programs:
\begin{align*}
  \text{(P)}\qquad \max && \alpha &&\\
  \text{s.t.} && \sum_{S \in \calI'} \frac{w(S_k)}{\OPT_k} x_S & \geq \alpha & \forall k \in K \\
  && \sum_{S \in \calI'} x_S & = 1 & \\
  && x_S & \geq 0 & \forall S \in \calI'
\end{align*}
\begin{align*}
  \text{(D)}\qquad \min && \beta &&\\
  \text{s.t.} && \sum_{k \in K} \frac{w(S_k)}{\OPT_k} y_k & \geq \beta & \forall S \in \calI' \\
  && \sum_{k \in K} y_k & = 1 & \\
  && y_k & \geq 0 & \forall k \in K
\end{align*}

Note that setting $x_S = \lambda_S$ for $S \in \calI'$ yields a feasible solution to (P) of value $\alpha$. 
For $k \in K$, let $y_k := \frac{\mu_k \OPT_k}{\sum_{i \in K} \mu_i \OPT_i}$ and let $\beta := \min_{S \in \calI'}  \sum_{k \in K} \frac{w(S_k)}{\OPT_k} y_k$. Note that $(y, \beta)$ is a feasible solution to (D) and hence $\beta \geq \alpha$. Therefore, there is $S \in \calI'$ such that 
\begin{align*}
\sum_{k \in K} w(S_k) \mu_k \; & = \; \sum_{k \in K} \frac{w(S_k)}{\OPT_k} y_k \sum_{i \in K} \mu_i \OPT_i \\
 \; & \geq \; \alpha \sum_{k \in K} \mu_k \OPT_k \; \geq \; \alpha \cdot \max_{S^* \in \calI} \, \sum_{k \in K} \mu_k w(S^*_k).\qedhere
\end{align*}
\end{proof}

Combining Theorem~\ref{thm:max-priority} with our results from the previous section, we obtain a $1/\ln(4)$-approximation algorithm for priority maximization over bit-concave independence systems.

\begin{corollary}
  There is a $1/\ln(4)$-approximation algorithm for finding a maximum priority independent set in bit-concave independence systems. In particular, there is a $1/ln(4)$-approximation for maximum priority matching.
\end{corollary}

\begin{remark}\label{rem:best-response}
Note that the proof of Theorem~\ref{thm:max-priority} does not make any assumptions on $\calI$ and $K$ other than that they are finite sets.
Theorem~\ref{thm:max-priority} therefore generalizes to arbitrary two-player zero-sum games as follows. 
Consider a zero-sum game played by Alice and Bob, where Alice has the finite strategy space $A$ and Bob has the finite strategy space $B$. Let $u(a, b)$ be the pay-off when Alice plays $a \in A$ and Bob plays $b \in B$.
Let $\lambda \in \Delta(A)$ be a mixed strategy of Alice with $\E_{a \sim \lambda}[u(a, b)] \geq \alpha \max_{a^* \in A} u(a, b)$ for every $b \in B$ (one could call this an \emph{$\alpha$-robust mixed strategy} of Alice). Theorem~\ref{thm:max-priority} then guarantees that for any mixed strategy $\mu \in \Delta(B)$ of Bob, $\lambda$ contains a pure strategy $a$ in its support with $\E_{b \sim \mu} u(a, b) \geq \alpha \max_{a^* \in A} \E_{b \sim \mu} u(a, b)$, i.e., the pure strategy $a$ is an $\alpha$-approximate best response to the mixed strategy $\mu$. 
\end{remark}

\subsubsection*{Complexity of finding optimal mixed strategies for Alice and Bob}  
The proof of Theorem~\ref{thm:max-priority} reveals another interesting connection between robust randomized matchings and maximum priority matching. Observe that finding a randomized matching of maximum robustness, \ie, an optimal mixed strategy for Alice, is equivalent to solving the linear program (P) with $\calI' = \calI$. Also note that the dual program (D) asks for an optimal mixed strategy played by Bob. As the number of variables in (P) is exponential, one could hope to solve the separation problem of the dual (D) instead, which reduces to finding~$S \in \calI$ minimizing $\sum_{k \in K} \frac{w(S_k)}{\OPT_k} y_k$ for a given $y \in \Delta(K)$. 
In game-theoretic terms, this can be interpreted as finding an optimal deterministic response by Alice to a given mixed strategy played by Bob.
In the proof of Theorem~\ref{thm:max-priority} we showed that this problem corresponds to the maximum priority matching problem (when choosing~$y$ as described in the proof), which is $\NP$-hard~\cite{hassin2002robust}. 
While this hardness effectively prevents us from solving the dual separation problem, we remark that it does not imply that finding the optimal strategies is $\NP$-hard: the equivalence of optimization and separation does not come into effect here, as (D) only contains the variable $\beta$ in its objective function. 

\subsubsection*{Original formulation of maximum priority matching}
At the beginning of this section, we introduced the maximum priority matching problem as a stochastic optimization variant of robust matchings. The original definition~\cite{hassin2002robust} of the problem differs somewhat from this formulation: The input is a graph $G = (V, E)$, edge weights $w : E \rightarrow \reals\nn$, and \emph{priorities} $c_1, \dots, c_{|E|} \in \reals\nn$ with~\mbox{$c_1 \geq \dots \geq c_{|E|}$}. The task is to find a matching $M = \{e_1, \dots, e_{|M|}\}$ with $w_{e_1} \geq \dots \geq w_{e_{|M|}}$ maximizing~\mbox{$f(M) := \sum_{i = 1}^{|M|} c_i w_{e_i}$}. It is, however, easy to see that this formulation is equivalent to the one given at the beginning of the section:
\begin{itemize}
  \item Given $\mu \in \Delta(\nats)$, set $c_k := \sum_{i = k}^{\infty} \mu_k$. Then $f(M) = \sum_{k \in \nats} \mu_k w(M_k)$ for every $M \in \calM$.
  \item Given $c_1, \dots, c_{|E|}$, set $\mu_{|E|} := c_{|E|}/c_1$ and $\mu_k := (c_k - c_{k+1})/c_1$ for $k \in \{1, \dots, |E|-1\}$. Then $\mu \in \Delta(\nats)$ and $f(M) = c_1 \cdot \sum_{k \in \nats} \mu_k w(M_k)$ for every $M \in \calM$.
\end{itemize}

\section{Asymptotic robustness}\label{sec:asymptotic}

	We now turn our attention to a setting in which Alice has to choose a matching deterministically, but Bob's choice of $k$ is limited to large cardinalities. Given $K \in \nats$, a matching $M$ is \emph{$(\alpha, K)$-robust}, if $w(M_k) \geq \alpha\OPT_k$ for every $k \geq K$. Accordingly, we define $\alpha_K(G)$ to be the \emph{$K$-asymptotic robustness ratio} of a graph $G = (V, E)$, \ie, the largest value $\alpha$ such that for all weights $w : E \rightarrow \reals\nn$ there is an $(\alpha, K)$-robust matching. Finally, we define $\alpha_K = \inf_{G} \alpha_K(G)$, \ie, $\alpha_K$ denotes the smallest $K$-asymptotic robustness ratio over all graphs.
	
	In this section, we establish a close connection between the concepts of randomized robustness and asymptotic robustness by showing that both the upper and lower bound on the randomized robustness factor given in the previous section carry over to the asymptotic setting.

	\begin{theorem}\label{thm:asymptotic}
		For every $\eps > 0$ there is a $K \in \nats$ such that $\alpha_K \geq 1/\ln(4) - \eps$.
	\end{theorem}

	The main idea in the proof of Theorem~\ref{thm:asymptotic} is to transform an $\alpha$-robust randomized matching into a deterministic $(\alpha - \eps, K)$-robust matching.
	Our argument, however, requires that the number of matchings in the support of the distribution is bounded by a constant. 
	We show that this property can be attained for the $1/\ln(4)$-robust distributions constructed in Section~\ref{sec:randomized} with only a small loss in the robustness factor; see Lemma~\ref{lem:constant-support}. If a similar sparsification exists for arbitrary randomized matchings, this would imply the following stronger result, which we leave as a conjecture:
	\begin{conjecture}\label{con:asymptotic}
	  $\lim_{K \rightarrow \infty} \alpha_{K} = \inf_{G} \alpha^*(G)$
	\end{conjecture}
	
\subsection{An upper bound on asymptotic robustness}
Before we discuss the proof of Theorem~\ref{thm:asymptotic}, we remark that our worst-case example for the robustness of randomized matchings can be translated to the asymptotic setting by introducing $K$ copies of the worst-case instance.

	\begin{lemma}\label{rem:asymptotic_bad_example}
		For every $K \in \nats$ there exists a graph $G$ such that $\alpha_K(G) \leq (1 + 1/\sqrt{2})/2$.
	\end{lemma}
	\begin{proof}
Consider the graph consisting of $K$ disjoint copies of the graph described in Figure~\ref{fig:tight}. Let $M$ be any maximal matching in this graph. Observe that in every copy of the original graph, $M$ either consists of the central edge with weight $\sqrt{2}$ or of the two outer edge with weight $1$ each. Let $K'$ denote the number of copies in which $M$ contains the central edge. Then 
\begin{align*}
& \frac{w(M_K)}{\OPT_K} = \frac{\sqrt{2}K'  + K - K'}{\sqrt{2}K}\\ \text{ and } \quad & \frac{w(M_{2K})}{\OPT_{2K}} = \frac{\sqrt{2}K'  + 2(K - K')}{2K}.
\end{align*}
Therefore the $K$-asymptotic robustness of $M$ is bounded by 
\begin{align*}
\max_{K' \in [0, K]} \min\left\{\frac{\sqrt{2}K'  + K - K'}{\sqrt{2}K},\ \frac{\sqrt{2}K'  + 2(K - K')}{2K}\right\},
\end{align*}
which is equal to $(1 + 1/\sqrt{2})/2$.
\end{proof}

	\subsection{Proof of Theorem~\ref{thm:asymptotic}}
	In order to prove Theorem~\ref{thm:asymptotic}, we first observe that the randomized matching produced by our algorithm in Section~\ref{sec:randomized} can be transformed into a distribution containing only a constant number of matchings in its support with only a small loss in its robustness factor.
	
	\begin{lemma}\label{lem:constant-support}
		For every $\eps > 0$, there is an $L \in \nats$ such that for every graph $G = (V, E)$ and every weight function $w : E \rightarrow \reals\nn$ there is an $(1/\ln(4)-\eps)$-robust randomized matching $\lambda \in \Delta(\mathcal{M})$  with support of size $|\{M \in \mathcal{M} \, : \, \lambda(M) > 0\}| \leq L$.
	\end{lemma}
	
\begin{proof}
Let $L \in \nats$ be sufficiently large such that \mbox{$2^{-1/L} \geq (1 - \eps)$}.
For $e \in E$, define \mbox{$\tilde{\ell}_e = \lfloor L \log_2 w_e \rfloor / L$} and $\tilde{w}_e = 2^{\tilde{\ell}_e}$. Note that $w_e \geq \tilde{w}_e \geq 2^{-1/L} w_e$. Therefore, any $\alpha$-robust randomized matching with respect to $\tilde{w}$ is $(1 - \eps)\alpha$-robust with respect to $w$. Now use the algorithm in the proof of Theorem~\ref{thm:ln4-robust} to construct a randomized $1/\ln(4)$-robust matching \mbox{$\lambda \in \Delta(\calM)$} with respect to $\tilde{w}$. Note that the rounded weight of edge~$e$ in the algorithm is~$\tilde{w}^{[x]}_e=2^{\tilde{\ell}_e}$ if~$x \leq \tilde{\ell}_e - \lfloor \tilde{\ell}_e \rfloor$ and~\mbox{$\tilde{w}^{[x]}_e=2^{\tilde{\ell}_e-1}$} otherwise. As all values $\tilde{\ell}_e$ are integer multiples of $1/L$, the algorithm considers at most $L$ different bit functions and thus~\mbox{$|\{M \in \mathcal{M} \, : \, \lambda(M) > 0\}| \leq L$}.
\end{proof}
	
	We then make use of the following lemma which describes how to merge two matchings from the support of the distribution according to their probability coefficients in such a way that the asymptotic robustness of the resulting matching is close to the randomized robustness factor of the convex combination of the two matchings. As we can assume the support of the distribution to be constant, we only need to repeat this procedure a constant number of times to construct a single matching that for large cardinalities is as good as the original randomized matching. In each step, the weight only decreases by a factor of $1 - \eps$, thus the total decrease is $(1 - \eps)^L$ for the chosen constant $L$. We can therefore bound the total loss by an arbitrarily small constant, proving Theorem~\ref{thm:asymptotic}.

	\begin{lemma}\label{lem:asymptotic_merge}
		For every $\varepsilon > 0$, there is a $K \in \mathbb{N}$ such that for every graph $G = (V, E)$, every weight function $w : E \rightarrow \mathbb{Z}_+$, every pair of matchings $M, M' \subseteq E$ in $G$, and every $\mu \in [0, 1]$, there is a matching $M^*$ such that $$w(M^*_k) \geq (1 - \varepsilon) \big(\mu w(M_k) + (1-\mu) w(M'_k)\big)$$ for all $k \geq K$.
	\end{lemma}

\subsubsection*{Outline of the proof}
We first give a short outline of the proof of Lemma~\ref{lem:asymptotic_merge}. We refer to an inclusionwise maximal cycle or path in the symmetric difference $M \Delta M'$ as a \emph{component}. The main idea of the proof is to construct $M^*$ by selecting from each component independently at random either all edges of $M$ or all edges of $M'$ according to the coefficient in the convex combination. However, in order to ensure that this procedure yields a matching satisfying the requirements of the lemma with positive probability, we need to  perform a series of transformations on $M$ and $M'$ that simplify the structure of the symmetric difference. 
	  
 The first transformation ensures that neighboring edges in $M \Delta M'$ do not differ too much in their weight by erasing edges of considerably smaller weight than their neighbors. The second ensures that the cardinality of each component of $M \Delta M'$ is not too large; this is done by partitioning each component into subpaths of constant lengths and erasing the minimum weight edge from each of them. The third and final transformation ensures that no component contributes more than a small fraction to the total weight of $M_K$ or $M'_K$, respectively. This last step is based on a consequence of the two preceding transformations: If a component $C$ contributes a too large fraction to the total weight of either $M_K$ or $M'_K$, then $C \subseteq M_K \cup M'_K$, \ie, all edges in $C$ are among the top $K$ edges of their respective matching. Therefore, we can swap the component in one of the matchings, thus eliminating it from the symmetric difference.
	   Note that performing these transformations can decrease the weights of $M_k$ and $M'_k$ for some $k \in \mathbb{N}$, but we can make sure that the decrease is bounded by a small fraction of the original weight. 
	   
	   The result of the above transformations is formalized in Lemma~\ref{lem:transformed} below. In particular, $M \Delta M'$ now consists of a large number of components all of which are small both in cardinality and weight. We can therefore apply Hoeffding's inequality~\cite{hoeffding1963probability} to ensure that with positive probability the constructed matching $M^*$ contains at most $(1 + \delta)k$ edges from $M_k \cup M'_k$ for every $k \geq K$, while the weight of $M^* \cap (M_k \cup M'_k)$ is at least $(1 - \delta)$ times the corresponding weight of the convex combination. Therefore, for this choice of $M^*$, there is a subset of $k$ edges of $M^*$ of weight $\frac{1-\delta}{1+\delta}(\mu w(M_k) + (1-\mu)w(M'_k))$ for every $k \geq K$, completing the proof.

\subsubsection*{Proof details.}
In the following, we describe the above steps in detail.
Choose $\delta > 0$ such that $\frac{(1 - \delta)^2}{1+\delta} \geq 1 - \eps$. Let $D_1 := 18/\delta + 3$, $D_2 := D_1^{D_1 + 4}$, and $D_3 := 6/\delta + 2$ (the important property is that these are sufficiently large numbers only depending on $\delta$). Throughout the proof, we consider a fixed $K$. We will determine the exact choice of $K$ at the end of the proof, but its value will only depend on $\delta$, which in turn only depends on $\eps$. The following lemma holds for any choice of $K$ and formalizes the outcome of the transformations outlined above:
	  
\begin{lemma}\label{lem:transformed}
  There are matchings $\bar{M}, \bar{M}' \subseteq M \cup M'$ such that
  \begin{itemize}
    \item $w(\bar{M}_k) \geq (1 - \delta) w(M_k)$ and $w(\bar{M}'_k) \geq (1 - \delta) w(M'_k)$ for all $k \geq K$,
    \item $|C| \leq D_1$ and $w(C) \leq \frac{D_2}{K} \min \{w(\bar{M}_K), w(\bar{M}'_K)\}$ for every component $C$ of $\bar{M} \Delta \bar{M}'$,
    \item $D_3 \cdot w(\bar{M}_k) \geq w(\bar{M}'_k) \geq w(\bar{M}_k) / D_3$ for all $k \geq K$.
  \end{itemize}
\end{lemma}

We first prove the lemma and then show how to use it to construct $M^*$.
	  
\subsubsection*{Initial transformations (Proof of Lemma~\ref{lem:transformed})}
We describe how to modify $M$ and $M'$ step by step so that the resulting matchings fulfill the properties claimed by the lemma.
Throughout the proof let $\delta' := \delta/3$.

We first perform the following operation on all edges in $M \Delta M'$ in non-increasing order of their weight, starting with the heaviest: Let $e$ be the current edge (assume $e \in M \setminus M'$; the case $e \in M' \setminus M$ is analogous). Let $e', e'' \in M'$ be the neighbors of $e$ in $M \Delta M'$ (if either of them does not exist, temporally introduce an edge of weight $0$ to some artificial vertex). Without loss of generality, we assume $w_{e'} \leq w_{e''}$. 
		\begin{enumerate}
			\item If $w_{e} \geq w_{e'} + w_{e''}$, replace $M'$ by $M' \setminus \{e', e''\} \cup \{e\}$.
			\item If $w_{e'} + w_{e''} > w_{e}$ but $w_{e'} \leq \frac{\delta'}{1 + \delta'} w_{e}$, then remove $e'$ from $M'$.
		\end{enumerate}
		Note that operations of the first type do not decrease $w(M'_k)$ for any $k$. Furthermore, if an operation of the second type is performed, then $w_{e''} \geq (1 - \frac{\delta'}{1 + \delta'}) w_{e} \geq \frac{1}{\delta'} w_{e'}$. As the edges of $M$ are processed in order of non-increasing weight, $e''$ will also not be removed from $M'$. Therefore $w(M'_k)$ does not decrease by more than a factor of $1 - \delta'$ in total. Finally, observe that after the transformation any neighboring pair of edges $e, e' \in M \Delta M'$ fulfills $w_{e'} > \frac{\delta'}{1 + \delta'} w_{e}$.
		
		The second transformation considers any component $C$ of $M \Delta M'$ of cardinality $|C| > 6/\delta' + 3$. It partitions $C$ into $\ell + 1$ subpaths $C_0, \dots, C_\ell$ such that $|C_i| = \lceil 2/\delta' \rceil$ for $i > 0$ and $|C_0| \leq \lceil 2 / \delta' \rceil$. For $i$ from $1$ to $\ell$, remove the minimum weight edge of $C_i$ from the corresponding matching. Note that for any $k \in \mathbb{N}$, this transformation decreases $w(M_k)$ and $w(M'_k)$ by a factor of at most $1 - \delta'$ as for any edge removed from either matching, $\lceil 1/ \delta' \rceil - 1$ edges of at least the same weight remain in the same matching. Furthermore, after the transformation is performed, any component of $M \Delta M'$ has length at most $3 \lceil 2 / \delta' \rceil \leq 6 / \delta' + 3 = D_1$ (in the worst case $C_0$ remained connected with two neighboring subpaths).
		
		Before we describe the third transformation, we establish some useful consequences of the first two transformations. To this end, we define $$\beta := \left(\frac{\delta'}{1 + \delta'}\right)^{D_1} \quad \text{ and } \quad \gamma := \frac{\beta}{D_1 D_3}.$$
				\begin{lemma}\label{lem:iterations_weight}
		  For every $k \in \nats$ it holds that $D_3 \cdot w(M_k) \geq w(M'_k) \geq w(M_k) / D_3$.
		\end{lemma}
		
		\begin{proof}
			After the first transformation, any edge $e$ has a neighbor in $M \Delta M'$ with weight at least $\delta'/(1 + \delta') w_{e} = \frac{2}{D_3} \cdot w_e$. The lemma follows from the fact that at most two edges share the same neighbor.
	  \end{proof}
		
		\begin{lemma}\label{lem:component_weight}
			Let $C$ be a component of $M \Delta M'$. Then $\min_{e \in C} w_{e} \geq \beta \cdot \max_{e \in C} w_{e}$.
		\end{lemma}
		
		\begin{proof}
		  After the second transformation, any component has length at most $D_1$. Therefore, the length of a path from the maximum weight edge to the minimum weight edge in $C$ is bounded by the same number. By the first transformation, with every edge along that path, the weight can only decrease by a factor of $\delta'/(1 + \delta')$. Therefore the total decrease is bounded by the factor~$\left(\delta'/(1 + \delta')\right)^{D_1} = \beta$.
    \end{proof}
		
		\begin{lemma}\label{lem:heavy_components}
		  If $C$ is a component with $w(C \cap M_K) > \frac{1}{\gamma K} w(M_K)$ or $w(C \cap M'_K) > \frac{1}{\gamma K} w(M'_K)$, then $C \subseteq M_K \cup M'_K$.
		\end{lemma}
		
		\begin{proof}
			We prove the lemma for the case $w(C \cap M_K) > \frac{1}{\gamma K} w(M_K)$. The case for $M'_K$ follows by symmetry. We define $w^* := \min_{e \in C} w_{e}$.
			By contradiction assume there is an edge $e \in C \setminus (M_K \cup M'_K)$. If $e \in M \setminus M_K$, then $M_K$ contains $K$ edges of weight at least $w_{e} \geq w^*$ and therefore $w(M_K) \geq Kw^*$ in this case. If $e \in M' \setminus M'_K$, then $M'_K$ contains $K$ edges of weight at least $w_{e} \geq w^*$. Therefore, $w(M_K) \geq w(M'_K) / D_3 \geq \frac{K}{D_3}w^*$ by Lemma~\ref{lem:iterations_weight}. We deduce that
			\[\frac{D_1}{\beta} w^* \ \geq \ w(C \cap M'_K) \ > \ \frac{1}{\gamma K} w(M_K) \ \geq \ \frac{1}{\gamma D_3} w^*\]
			where the first inequality follows from the fact that every component contains at most $D_1$ edges and that the maximum weight edge in $C$ has weight at most $w^*/\beta$ by Lemma~\ref{lem:component_weight}. The above inequality implies that 
		$\gamma > \frac{\beta}{D_1 D_3}$, a contradiction to the definition of~$\gamma$.
\end{proof}

    Our third and final transformation will eliminate components from the symmetric difference that contribute a too large fraction to the weight of $M_K$ or $M'_K$.     
		The transformation proceeds as long as there is a component $C$ with $w(C \cap M_K) > \frac{1}{\gamma K} w(M_K)$ or $w(C \cap M'_K) > \frac{1}{\gamma K} w(M'_K)$ and performs the following operation: If $w(C \cap M) \geq w(C \cap M')$ then replace $M'$ by $M' \Delta C$. Otherwise, replace $M$ by $M \Delta C$. Note that by doing so, the component $C$ is removed from $M \Delta M'$.
		
		We show that, in total, this final transformation decreases $w(M_k)$ and $w(M'_k)$ for $k \geq K$ by at most a factor of $1 - 2\gamma$. 
		Consider the iteration of the transformation in which it operates on a component $C$. If $w(C \cap M) \geq w(C \cap M')$, then $M$ is not modified at all. Hence consider the case  $w(C \cap M) < w(C \cap M')$.	Note that, before the transformation, by Lemma~\ref{lem:heavy_components}, $C \cap M_k = C \cap M$ and $C \cap M'_k = C \cap M'$ for all $k \geq K$ and that further $|C \cap M'| \leq |C \cap M| + 1$. Therefore $|M_k \Delta C| \leq k + 1$. Thus, the operation decreases $w(M_k)$ by at most the weight of its minimum weight edge. As there are at most $2 \gamma K$ components with $w(C \cap M_K) > \frac{1}{\gamma K} w(M_K)$ or $w(C \cap M'_K) > \frac{1}{\gamma K} w(M'_K)$, the weight of $M_k$ decreases by at most $ \frac{2 \gamma K}{k} w(M_k) \leq 2 \gamma w(M_k)$. Analogously, $w(M'_k)$ decreases by at most $2 \gamma w(M'_k)$.

		Finally, to conclude the proof of Lemma~\ref{lem:transformed}, we observe that the total decrease incurred by the transformations to $w(M_k)$ or $w(M'_k)$ for any $k \geq K$ is not worse than a factor~$(1 - \delta')^2 (1 - 2\gamma) \geq (1 - \delta')^3 \geq 1 - \delta$. Furthermore, we note that Lemma~\ref{lem:heavy_components} implies that the heaviest edge appearing in $M \Delta M'$ after the third transformation has a weight of at most $\frac{1}{\gamma K} \max \{w(M_K), w(M'_K)\} \leq \frac{D_3}{\gamma K} \min \{w(M_K), w(M'_K)\}$. Hence no component can have a weight larger than $\frac{D_1 D_3}{\gamma K} \cdot \min \{w(M_K), w(M'_K)\} \leq \frac{D_2}{K} \min \{w(M_K), w(M'_K)\}$ (note that $D_3 \leq D_1$ and $\gamma \geq D_1^{-(D_1 + 2)}$).

\subsubsection*{Construction of $M^*$ (Proof of Lemma~\ref{lem:asymptotic_merge})}
		
In the following, let $\bar{M}$ and $\bar{M}'$ the matchings obtained by applying Lemma~\ref{lem:transformed} to $M$ and $M'$.
Let $\mathcal{C}$ be the set of components in $\bar{M} \Delta \bar{M}'$.
We construct a random matching $M^*$ by choosing independently for each component $C \in \mathcal{C}$ the edges of $C \cap \bar{M}$ with probability $\mu$ and the edges of $C \cap \bar{M}'$ with probability $1 - \mu$. We further add to $M^*$ all edges in $\bar{M} \cap \bar{M}'$. We will show that $M^*$ simultaneously fulfills all requirements of Lemma~\ref{lem:asymptotic_merge} with positive probability, concluding the proof. To this end, we will use Hoeffding's inequality.

\begin{theorem}[Hoeffding~(1963)~\cite{hoeffding1963probability}]\label{thm:hoeffding}
  Let $X_1, \dots, X_n$ be random variables with $X_i \in [a_i, b_i]$ almost surely for every $i \in \{1, \dots, n\}$. Define $X := \sum_{i=1}^{n} X_i$ and let $t \geq 0$. Then 
    $$\operatorname{Pr}\big(X - \mathbb{E}[X] \,\geq\, t\big) \ \leq \ \exp\left(-\frac{2 t^2}{\sum_{i=1}^{n}(b_i - a_i)^2}\right).$$
\end{theorem}
		
		We first ensure that $M^*$ contains not considerably more than $k$ edges from $\bar{M}_k \cup \bar{M}'_k$ for any $k \geq K$. To this end, we define $n_{C, k} := |C \cap \bar{M}_k|$ and $n'_{C, k} := |C \cap \bar{M}'_k|$ for all $C \in \mathcal{C}$ and $k \geq K$. We further define the random variable $N^*_k := |M^* \cap (\bar{M}_k \cup \bar{M}'_k)|$ for every $k \geq K$.
		Observe  that $\mathbb{E}[N^*_k] = k$ and thus by Hoeffding's inequality
		\begin{align*}
		  \operatorname{Pr}\big(N^*_k \geq (1 + \delta)k\big) \ \leq \ &  \exp\left(-\frac{2\delta^2k^2}{\sum_{C \in \mathcal{C}} (n_{C, k} - n_{C, k})^2}\right) \\
		  \leq \ & \exp\left(-\frac{\delta^2k}{D_1}\right)
		\end{align*}
		for every $k \geq K$. Here, the second inequality follows from the fact that $\sum_{C \in \mathcal{C}} |n_{C, k} - n_{C, k}| \leq 2k$ and $|C| \leq D_1$ for all $C \in \mathcal{C}$. Defining $r := \exp\left(-\frac{\delta^2}{D_1}\right)$, we deduce that the probability that $N^*_k \geq (1 + \delta) k$ for at least one $k \geq K$ is at most
		\[\sum_{k \geq K}  \operatorname{Pr}\big(N^*_k \geq (1 + \delta)k\big) \ \leq \ \sum_{k \geq K} r^k \ \leq \ \frac{r^K}{1 - r}.\]
		
		We now turn our attention to the weights. For $k \in \mathbb{N}$, we define $W_k := \mu w(\bar{M}_k) + (1 - \mu) w(\bar{M}'_k)$ and for $C \in \mathcal{C}$ we define $w_{C,k} := w(C \cap \bar{M}_k)$ and $w'_{C,k} := w(C \cap \bar{M}'_k)$. 
		Observe that
		\begin{align*}
		  |w_{C,k} - w'_{C,k}| \ \leq \ \frac{D_2}{K} \min \{w(\bar{M}_K), \, w(\bar{M}'_K)\} \ \leq \ \frac{D_2}{K} W_K
		\end{align*}
		by Lemma~\ref{lem:transformed}.
		Furthermore,		
		\begin{align*}
		  \sum_{C \in \mathcal{C}} |w_{C,k} - w'_{C,k}| \ \leq \ & w(\bar{M}_k) + w(\bar{M}'_k) \ \leq \ 2 D_3 \cdot \min \{w(\bar{M}_k), \, w(\bar{M}'_k)\} \ \leq \ 2 D_3 \cdot W_k
		\end{align*}
		by Lemma~\ref{lem:transformed}. This implies that 
		\begin{align*}
		  \sum_{C \in \mathcal{C}} (w_{C,k} - w'_{C,k})^2 \ \leq \ & \frac{2D_3 W_k}{\frac{D_2}{K} W_K} \left(\frac{D_2}{K} W_K\right)^2 
		  \ \leq \ \frac{2 D_2 D_3}{K} W_k W_K.
		\end{align*}
		
		For $k \geq K$ we introduce the random variable $W^*_k := w(M^* \cap (\bar{M}_k \cup \bar{M}'_k))$. Note that $\mathbb{E}[W^*_k]= W_k$. Again using Hoeffding's inequality and then applying the above observations, we deduce that
		\begin{align*}
		  \operatorname{Pr}\left(W^*_k \leq \left(1 - \frac{\delta}{2}\right) W_k \right) \ \leq \ &  \exp\left(-\frac{\frac{1}{2}\delta^2 W_k^2}{\sum_{C \in \mathcal{C}} (w_{C,k} - w'_{C,k})^2}\right) \\
		  \leq \ & \exp\left(-\frac{\delta^2 K W_k} {4 D_2 D_3 W_K}\right).
		\end{align*}
		Note that in order to guarantee that $W^*_k \geq (1 - \delta)W_k$ for all $k \geq K$, it suffices to check $W^*_k \leq (1 - \delta/2)W_k$ once at the end of every interval in which $W_k$ increases by a factor of at most $1 + \delta/2$. Formally, we define $k_0 = K$ and $k_{i + 1} := \min \{k > k_i \, : \, W_k \leq (1 + \delta/2) W_{k_i} \}$, stopping with $k_\ell = \infty$ for some $\ell \in \mathbb{N}$.
		Observe that if $W^*_{k_i} \geq (1 - \delta/2)W_{k_i}$, then 
		\[W^*_k \geq W^*_{k_i} \geq (1 - \delta/2)W_{k_i} \geq \frac{1 - \delta/2}{1 + \delta/2}W_k \geq (1 - \delta) W_k\]
		for all $k$ with $k_i \leq k < k_{i+1}$.
		Defining $q := \exp\left(-\frac{\delta^2} {4 D_2 D_3}\right)$, we therefore bound the probability
		\begin{align*}
		  \sum_{i = 0}^{\ell-1} \operatorname{Pr}\left(W^*_{k_i} \leq \left(1 - \frac{\delta}{2}\right) W_{k_i} \right) \ \leq \ & \sum_{i = 0}^{\infty} q^{\frac{(1 + \delta/2)^i W_K}{W_K} K} \ = \ q^K \underbrace{\sum_{i = 0}^{\infty} q^{(1 + \delta/2)^i}}_{< \infty}.
		  \end{align*}
		  
		  Observe that by choosing $K$ sufficiently large such that 
		  \[\frac{r^K}{1 - r} + q^K \sum_{i = 0}^{\infty} q^{(1 + \delta/2)^i} < 1\]
		  we can ensure that $N^*_k \leq (1 + \delta)k$ and $W^*_k \geq (1  - \delta)W_k$ for all $k \geq K$ with positive probability. This implies that there exists a matching $M^*$ that for every $k \geq K$ contains a subset of cardinality at most $(1 + \delta)k$ with a total weight of $(1 - \delta)W_k$. Choosing the top $k$ edges from this subset guarantees that
  \begin{align*}
		w(M^*_k) & \ \geq \ \frac{1-\delta}{1+\delta}W_k \ = \ \frac{1-\delta}{1+\delta} \left(\mu w(\bar{M}_k) + (1 - \mu) w(\bar{M}'_k)\right)\\
		 & \ \geq \ \underbrace{\frac{1-\delta}{1+\delta} (1 - \delta)}_{\geq 1 - \eps} \left(\mu w(M_k) + (1 - \mu) w(M'_k)\right).
	\end{align*} 
  This concludes the proof of Lemma~\ref{lem:asymptotic_merge}.

\section{An LP-based proof for the squared weights algorithm}\label{sec:sqrt2-proof}

Our final contribution is a new LP-based proof for the $1/\sqrt{2}$-robustness of the matching computed by the squared weights algorithm.

\begin{theorem}[Hassin, Rubinstein (2002)~\cite{hassin2002robust}]\label{thm:squared_weights}
Let $M \subseteq E$ be a matching maximizing the squared weights $\sum_{e\in M} w^2_e$. Then $M$ is $1/\sqrt{2}$-robust.
\end{theorem}

\noindent
In order to prove Theorem~\ref{thm:squared_weights}, consider a fixed cardinality $k \in \nats$ and let $M^*$ be the optimal $k$-matching. Observe that, without loss of generality, we can restrict to the case of a complete bipartite graph, as the subgraph of `interesting' edges $(V,\, M_k\cup M^*)$ is bipartite. Therefore the matching~$M$ is an optimal solution to the following primal linear program:
\begin{align*}
\max\quad & \sum_{e\in E} w_e^2x_e\\
\text{s.t.}\quad & \sum_{e\in\delta(v)}x_e\leq1 & \forall\, v\in V\\
& x_e\geq0 & \forall\, e\in E &
\end{align*}
We will also consider the corresponding dual:
\begin{align*}
\min\quad&\sum_{v\in V}y_v^2\\
\text{s.t.}\quad& y_u^2+y_v^2\geq w_{uv}^2 & \forall\, uv\in E\\	
& y_v^2\geq0 & \forall\, v\in V
\end{align*}
Notice that we can denote the dual variable of node~$v$ by $y_v^2$ as it is non-negative. 

The $k$-matching~$M^*$ is an optimal solution to the following primal linear program:
\begin{align*}
\max\quad & \sum_{e\in E} w_ex^*_e\\
\text{s.t.}\quad & \sum_{e\in\delta(v)}x^*_e\leq1 & \forall\, v\in V\\
& \sum_{e\in E}x_e^*\leq k &\\
& x_e^*\geq0 & \forall\, e\in E 
\end{align*}
Again, we will also consider the dual linear program:
\begin{align*}
\min\quad&k\cdot z^*+\sum_{v\in V}y_v^*\\
\text{s.t.}\quad& z^*+y_u^*+y_v^*\geq w_{uv} & \forall\, uv\in E\\	
& y_v^*\geq0 & \forall\, v\in V
\end{align*}

Notice that $M_k$ is a feasible primal solution. The idea of our proof is to turn an optimal dual solution $y^2$ to the first pair of LPs into a feasible dual solution $y^*,z^*$ to the second pair of LPs whose objective function value is at most $\sqrt{2}$ times the weight of~$M_k$.

We describe the construction of $y^*,z^*$. For simplicity we assume (by scaling) that the cheapest edge in~$M_k$ has weight~$1$. We set~$z^*:=\sqrt{2}$ and $y^*_v:=0$ for all nodes $v$ that are not incident to an edge in~$M_k$. Finally, we are going to set the remaining variables~$y^*_v$ for nodes~$v$ incident to an edge in~$M_k$ in such a way that~$y^*_u+y^*_v=\sqrt{2}(w_{uv}-1)$ for all~$uv\in M_k$. Consider an edge $uv\in M_k$ and let $y_u^2\geq y_v^2$ such that $y_u\geq y_v$. If $y_v<1/\sqrt{2}$, then set
\begin{align*}
y^*_v:=0\qquad\text{and}\qquad y^*_u:=\sqrt{2}(w_{uv}-1).
\end{align*}
Notice that~$y^*_u\geq0$ as $w_{uv}\geq1$. If $y_v\geq 1/\sqrt{2}$, then set
\begin{align*}
& y^*_u:=\sqrt{2}\left(\frac{w_{uv}}{y_u+y_v}y_u-\tfrac12\right)\\
\text{and} \qquad & y^*_v:=\sqrt{2}\left(\frac{w_{uv}}{y_u+y_v}y_v-\tfrac12\right).
\end{align*}
Notice that these values are non-negative since, by complementary slackness, $y_u^2+y_v^2=w_{uv}^2$ and therefore $w_{uv}/(y_u+y_v)\geq 1/\sqrt{2}$. In particular we get \mbox{$y_u^*\geq y_u-1/\sqrt{2}$} and $y_v^*\geq y_v-1/\sqrt{2}$ in this case.

\begin{lemma}\label{lem:value}
The value of the constructed dual solution~$y^*,z^*$ is equal to $\sqrt{2}$ times the total weight of edges in~$M_k$.	
\end{lemma}

\begin{proof}
This follows easily from the fact that there are exactly~$k$ edges in~$M_k$ and~$y^*_u+y^*_v=\sqrt{2}(w_{uv}-1)$ for all~$uv\in M_k$.
\end{proof}

\begin{lemma}\label{lem:feasible}
The constructed dual solution~$y^*,z^*$ is feasible.
\end{lemma}

\begin{proof}
We need to show that
$z^*+y_u^*+y_v^*\geq w_{uv}$ for all $uv\in E$.
As $z^*=\sqrt{2}$, the inequality trivially holds for all edges of weight at most~$\sqrt{2}$ and for all edges in~$M_k$. So let us consider an edge $vv'\in E\setminus M_k$ of weight~$w_{vv'}=\sqrt{2+\eps}$ for some $\eps>0$. 
We will make use of the fact that $y_v^2+y_{v'}^2\geq w_{vv'}^2$ by feasibility of the dual solution~$y^2$.

We assume that $y_v^*\geq y_{v'}^*$ and distinguish two cases. If $y_{v'}^*=0$, then $y_{v'}\leq 1$ by construction of~$y^*$ and complementary slackness (note that either $v'$ is incident to an edge $u'v' \in M_k$ with $w_{u'v'} = 1$ or it is covered by an edge in $M \setminus M_k$, all of which have weight at most $1$).
Since $y_v^2+y_{v'}^2\geq w_{vv'}^2$, we conclude that $y_v\geq\sqrt{1+\eps}$. Thus, there is an edge $uv\in M_k$. Then, $w_{uv} = \sqrt{y_v^2 + y_w^2} \geq y_v\geq\sqrt{1+\eps}$ and thus, by construction of~$y^*$,
\begin{align*}
z^*+y_v^*+y_{v'}^*&\geq \sqrt{2}+\min\bigl\{\sqrt{2}(w_{uv}-1),y_v-1/\sqrt{2}\bigr\}\\
&\geq\min\bigl\{\sqrt{2}\sqrt{1+\eps},\sqrt{1+\eps}+1/\sqrt{2}\bigr\}\\
&\geq\sqrt{2+\eps}=w_{vv'}.
\end{align*}
This concludes the first case. We now consider the case $y_v, \, y_{v'} > 0$ and let $uv, u'v' \in M_k$. We get 
	\begin{align*}
		z^* + y^*_v + y^*_{v'} \ \geq \ & \sqrt{2} + \min \left\{\sqrt{2}(w_{uv} - 1),\ y_v - 1/\sqrt{2}\right\}\\
		&  + \min \left\{\sqrt{2}(w_{u'v'} - 1),\ y_{v'} - 1/\sqrt{2}\right\}.
	\end{align*}
	We distinguish three cases for the possible combinations of the minima and show that in each case the right hand side value is at least $w_{vv'}$, implying feasibility of the solution.
	\begin{itemize}
		\item Assume both minima are attained by the first expression. Note that $$\min \{w_{uv}, \, w_{u'v'}\} \geq 1 \text{ and } \max \{w_{uv}, \, w_{u'v'}\} \geq w_{vv'}/\sqrt{2},$$ where the second inequality follows from $w_{uv}^2 + w_{u'v'}^2 \geq y_{v}^2 + y_{v'}^2 \geq w_{vv'}^2$. We thus get
			\begin{align*}
					z^* + y^*_v + y^*_{v'} \ =  & \ \sqrt{2}(w_{uv} + w_{u'v'} - 1)\\
					 \ \geq & \ \sqrt{2}(w_{vv'} / \sqrt{2} + 1 - 1) \ = \ w_{vv'}.
			\end{align*}
		\item If both minima are attained by the second expression, then 
			$z^* + y^*_v + y^*_{v'} \ = \ y_v + y_{v'} \ \geq\ w_{vv'}$.
		\item Finally, assume the first minimum is attained by the first expression and the second minimum is attained by the second expression (the remaining case follows by symmetry). We obtain 
		\begin{align*}
			z^* + y^*_v + y^*_{v'} \ =\ & \sqrt{2} w_{uv} + y_{v'} - 1/\sqrt{2}\\
		 \ \geq \ & \sqrt{2} \max\left\{1, \ \sqrt{w_{vv'}^2 - y_{v'}^2}\right\}\\
		 & \quad + y_{v'} - 1/\sqrt{2}
		\end{align*}
		as $w_{uv} \geq 1$ and $w_{uv}^2 + y_{v'}^2 \geq y_{v}^2 + y_{v'}^2 \geq w_{vv'}^2$. Note that $1 \geq \sqrt{w_{vv'}^2 - y_{v'}^2}$ if and only if $y_{v'} \geq \sqrt{w_{vv'}^2 - 1}$. We perform a final case distinction.
		\begin{itemize}
			\item If $y_{v'} \geq \sqrt{w_{vv'}^2 - 1}$, then 
				\begin{align*}
					z^* + y^*_v + y^*_{v'} \ \geq \ & \sqrt{2} + \sqrt{w_{vv'}^2 - 1} - 1/\sqrt{2},
				\end{align*}
				which is at least $w_{vv'}$ because $w_{vv'} > \sqrt{2}$.
			\item If $y_{v'} < \sqrt{w_{vv'}^2 - 1}$, then 
				\begin{align*}
					z^* + y^*_v + y^*_{v'} \ \geq \ & \sqrt{2 (w_{vv'}^2 - y_{v'}^2)} + y_{v'} - 1/\sqrt{2}
				\end{align*}
				Observe that the right hand side is concave in $y_{v'}$ and therefore its minimum for $y_{v'} \in [1/\sqrt{2}, \sqrt{w_{vv'}^2 - 1}]$ is attained at one of the endpoints of the interval.
				Note that the case $y_{v'} = \sqrt{w_{vv'}^2 - 1}$ has already been handled above. For $y_{v'} = 1/\sqrt{2}$ the right hand side becomes $\sqrt{2w_{vv'}^2 - 1} \geq w_{vv'}$, again because $w_{vv'} > \sqrt{2}$.
		\end{itemize} 
	\end{itemize}
This concludes the proof of the lemma.
\end{proof}

As Theorem~\ref{thm:squared_weights} follows immediately from Lemmas~\ref{lem:value} and~\ref{lem:feasible}, this concludes the proof.

\section{Conclusion and future research}

In this paper, we studied new variants of Hassin and Rubinstein's classic result on robust matchings. In particular we proved that a randomized matching with robustness $1/\ln(4)$ always exists and we showed that the randomized setting is closely related to the asymptotic setting, in which only large cardinalities are considered. 
While the analysis of our algorithm is tight, a better guarantee might still be derived from a different approach and closing the gap to the upper bound of $(1 + 1/\sqrt{2})/2$ on the best possible guarantee established by the worst-case instance remains as an interesting task for future research.

In Section~\ref{sec:max-priority}, we established a connection between the problem of finding optimal mixed strategies for Alice and Bob in the robust matching game and the maximum priority matching problem, a stochastic variant of matching with uncertain cardinality constraint. This connection provided interesting insights in two ways: Firstly, it allowed us to translate our $1/\ln(4)$-robustness for  randomized matchings into an approximation guarantee for maximum priority matchings. We expect that this connection between randomized strategies in robust optimization and the stochastic optimization variant can also be exploited for other optimization problems. Secondly, we observed that the $\NP$-hardness of maximum priority matching implies $\NP$-hardness for the separation problem associated with finding an optimal robust randomized matching. However, as we pointed out, this does not imply that the problem itself is $\NP$-hard, and hence the complexity of finding an optimal strategy for Alice remains an intriguing open question.

In Section~\ref{sec:asymptotic}, we showed how to transform the $1/\ln(4)$-robustness for randomized matchings into an asymptotic robustness guarantee for deterministic matchings.  For proving this result we used a probabilistic argument that shows the existence of the corresponding deterministic matching. It is open whether our construction can be derandomized  to obtain an efficient algorithm for computing the asymptotically robust matching. 
Finally, it would also be interesting to find out whether a guarantee for randomized robustness always translates to an asymptotic robustness guarantee, without conditions on the support of the distribution; see Conjecture~\ref{con:asymptotic}.

\section*{Acknowledgments.}
We thank two anonymous referees for helpful comments that in particular led to Remark~\ref{rem:best-response} and Conjecture~\ref{con:asymptotic}.
This work was supported by the German Academic Exchange Service~(DAAD) with funds of BMBF and the EU Marie Curie Actions, by Nucleo Milenio Informaci\'on y
Coordinaci\'on en Redes ICM/FIC RC130003, FONDECYT grant 11130266 and CONICYT PCI PII 20150140", by the Research Center \textsc{Matheon} in Berlin, and by the Einstein Foundation Berlin.

\bibliographystyle{plain}
\bibliography{references}

\section*{Appendix: Examples for strong $2$-exchange systems}

We complete our discussion of strong $2$-exchange systems by showing that this class of systems contains $b$-matchings, strongly base orderable matchoids and, more generally, strongly base orderable matroid parity systems. For completeness, we recall the definition of these systems.
Given a graph $G=(V,E)$ and a vector $b \in \nats^{V}$ of degree bounds, the set of \emph{$b$-matchings} of $G$ is the collection $\mathcal{M}_b(G)=\{F \subseteq E \colon \deg_{F}(v)\leq b_v \ \forall\, v \in V\}$. If instead of having a vector of degree bounds, we are given a collection of matroids $\{M_v\}_{v\in V}$, where the ground set of $M_v$ is the star $\delta_G(v)$ of $v$, the associated \emph{matchoid} is the independence system that consists of all $F \subseteq E$ such that $F\cap \delta(v)$ is independent in $M_v$ for all $v \in V(G)$. Finally, given a 1-regular graph $G=(V,E)$ and a matroid $M$ on $V$, the \emph{matroid parity system} $\calJ(G,E)$ is the system whose independent sets are the subsets $F\subseteq E$ of edges such that the set $\bigcup F \subseteq V$ of vertices covered by $F$ is independent in the matroid~$M$. We call a matchoid SBO if all the matroids $M_v$ are strongly base orderable\footnote{A matroid $M$ is strongly base orderable if for every pair of bases $U$ and $W$ of $M$ there is a \emph{base exchange} bijection $\pi\colon U \to W$  such that for all $Z\subseteq U, Z \cup \pi(U\setminus Z)$ is a base of $M$. Strongly base orderable matroids include partition matroids, transversal matroids and gammoids.}. Similarly, we say that  a matroid parity system is SBO if its associated matroid is strongly base orderable.

Note that $b$-matchings are a special case of SBO matchoids where each associated matroid is uniform, and thus strongly base orderable. It is also easy to see that every matchoid is a matroid parity system: Given a matchoid $\calI=\calI(G,(M_v)_{v\in V})$ we can construct a graph $G'$ by replacing each vertex of $v$ by $\deg_G(v)$ distinct copies, and replacing each edge $e=uv \in E(G)$ by an edge $e'$ connecting a copy of $u$ and a copy of $v$ in such a way that the new edges are pairwise disjoint. The resulting graph $G'=(V',E')$ is 1-regular. The matroid $M_v$, originally defined on the star $\delta_G(v)$, can be seen as a matroid $M'_v$ on the corresponding $\deg_G(v)$ copies of $v$ in $V'$ (where we associate each edge $e$ of $\delta_G(v)$ with the unique copy of $v$ contained in $e'$). Define a new matroid $M$ on ground set $V'$ by taking the union of all matroids $M'_v$. It is easy to see that $F\subseteq E$ is independent in $\calI$ if and only if $E'$ is independent in the matroid parity system $\calJ(G',M')$. Furthermore if all matroids $M_v$ are strongly base orderable, then so is the matroid $M'$. To see this, let $U$ and $W$ be two bases of $M'$  and let $U(v)$ and $W(v)$ for $v \in V$ be the collection of copies of $v$ included in $U$ and $W$ respectively. Then $U(v)$ and $W(v)$ are bases of the  strongly base orderable matroid $M'_v$, and thus there is a base exchange bijection $\pi_v$ from $U(v)$ to $W(v)$. By taking the union of all these bijections we obtain a base exchange bijection $\pi$ from $U$ to $W$.

By the previous paragraph, we only need to show the following result:

\begin{lemma}\label{thm:matchoid}
Let $\calJ(G,M)$ be the matroid parity system associated to the 1-regular graph $G=(V,E)$ and the strongly base orderable matroid $M$ on $V$, then $\calJ$ is  a strong 2-exchange system.
\end{lemma}

\begin{proof}
Let $X$ and $Y$ be two independent sets in $\calJ(G,M)$, and $s=||X|-|Y||$. Let $E_0=\{a_ib_i \colon 1 \leq i \leq s\}$ be a set of $s$ dummy edges outside $E$, where $V_0=\{a_1, \dots, a_s, b_1, \dots, b_s\}$ are $2s$ new dummy vertices. If $|X|<|Y|$, let $X'=(X\setminus Y)\cup E_0$ and $Y'=(Y\setminus X)$, otherwise, let $X'=(X\setminus Y)$ and $Y'=(Y\setminus X)\cup E_0$, in such a way that $|X'|=|Y'|=:r$.

Consider the matroid $M'$ on $V(G')$ obtained by first contracting $C:=\bigcup(X\cap Y)$, then adding elements $V_0$ to $M$ as coloops\footnote{An element $x$ is a coloop of a matroid if it is included in every basis.} and then truncating it to rank $2r$. It is easy to check that strong base orderability is closed under adding coloops, contractions, deletions and truncation (see, \eg, \cite{Schrijver} for an exposition). Therefore, the matroid $M'$ is also strongly base orderable and furthermore $U=\bigcup X'$ and $W=\bigcup Y'$ are bases of~$M'$.

Consider a base exchange bijection $\pi\colon U \to W$ for the matroid $M'$ from  $U$ to $W$, and let $H$ be the bipartite graph with color classes $X'$ and $Y'$ having an edge between $e=uv \in X'$ and $f=u'v'\in Y'$ if $\pi(\{u,v\}) \cap \{u',v'\} \neq \emptyset$. We claim that the subgraph $H'$ of $H$ obtained by deleting the dummy vertices and edges is a 2-exchange graph from $X$ to $Y$ in $\calJ(G,M)$.

Indeed, $H'$ is a bipartite graph with color classes $X\setminus Y$ and $Y\setminus X$, with maximum degree at most 2. Furthermore, let $A \subseteq X\setminus Y$ be a set of vertices in $H'$ (which are edges in $G$). We will  show that the set $F:=A \cup (Y\setminus N_{H'}(A))$ is independent in $\calJ(G,M)$.

Indeed, let $Z=\bigcup A \subseteq U$ be the vertices of $G$ covered by $A$. Then the set of vertices of $G$ covered by $F$ is
\begin{align*}
\bigcup F &\subseteq \bigcup ( (X\cap Y) \cup A \cup (Y' \setminus N_{H}(A)) \\
&\subseteq C \cup Z \cup \bigcup (Y' \setminus N_H(A))\\ 
&\subseteq C \cup Z \cup \pi(U\setminus Z).
\end{align*}

Since $Z\cup \pi(U\setminus Z)$ is a base of $M'$ we conclude that $\bigcup F$ is independent in $M$. Therefore, $F$ is independent in $\calJ(G,M)$.
\end{proof}

\end{document}